\documentclass{IEEEtran}
\usepackage{hyperref}
\usepackage{color, dsfont, bm, epsfig, graphicx, cite, multirow, multicol, comment}
\usepackage[ ruled]{algorithm2e}
\usepackage{amsmath,amsfonts,amsthm,amssymb,subfigure}
\usepackage{stfloats, booktabs,threeparttable, tabularx}
\usepackage{mathrsfs,bookmark}

\newcommand{\st}[1]{\text{s.t.}}

\newtheorem{thm}{Theorem}

\newtheorem{definition}{Definition}
\newtheorem{lemma}{Lemma}
\newtheorem{assumption}{Assumption}

\psfull
\begin{document} 
\title{Federated Fine-Tuning for Pre-Trained Foundation Models Over Wireless Networks}	 
\author{
	\normalsize{Zixin Wang, \textit{Member, IEEE}, Yong Zhou, \textit{Senior Member, IEEE}, \\ Yuanming Shi, \textit{Senior Member, IEEE}, and Khaled. B. Letaief, \textit{Fellow, IEEE}}
	\thanks{Part of this paper has been submitted to the IEEE Global Communications Conference,  Dec. 2024. Z. Wang and K. B. Letaief are with the Department of Electronic and Computer Engineering, The Hong Kong University of Science and Technology, Hong Kong (E-mail: \{eewangzx, eekhaled\}@ust.hk). Y. Zhou and Y. Shi are with the School of Information Science and Technology, ShanghaiTech University, Shanghai, 201210, China (E-mail: \{zhouyong, shiym\}@shanghaitech.edu.cn).
	}
}	
\maketitle
\begin{abstract}
Pre-trained foundation models (FMs), with extensive number of neurons, are key to advancing next-generation intelligence services, where personalizing these models requires massive amount of task-specific data and computational resources. The prevalent solution involves centralized processing at the edge server, which, however, raises privacy concerns due to the transmission of raw data.
Instead, federated fine-tuning (FedFT) is an emerging privacy-preserving fine-tuning (FT) paradigm for personalized pre-trained foundation models. 
In particular, by integrating low-rank adaptation (LoRA) with federated learning (FL), federated LoRA enables the collaborative FT of a global model with edge devices, achieving comparable learning performance to full FT while training fewer parameters over distributed data and preserving raw data privacy. 
However, the limited radio resources and computation capabilities of edge devices pose significant challenges for deploying federated LoRA over wireless networks. 
To this paper, we propose a split federated LoRA framework, which deploys the computationally-intensive encoder of a pre-trained model at the edge server, while keeping the embedding and task modules at the edge devices. The information exchanges between these modules occur over wireless networks. Building on this split framework, the paper provides a rigorous analysis of the upper bound of the convergence gap for the wireless federated LoRA system. This analysis reveals the weighted impact of the number of edge devices participating in FedFT over all rounds, motivating the formulation of a long-term upper bound minimization problem.  To address the long-term constraint, we decompose the formulated long-term mixed-integer programming (MIP) problem into sequential sub-problems using the Lyapunov technique. We then develop an online algorithm for effective device scheduling and bandwidth allocation. Simulation results demonstrate the effectiveness of the proposed online algorithm in enhancing learning performance.
\end{abstract}
\section{Introduction}\label{SecIntro}
The rapid advancements in artificial intelligence (AI), particularly in the development of pre-trained foundation models (FMs) such as large language models (LLMs) and large vision models (LVMs), have been truly remarkable. 
Applications like ChatGPT, DALL-E, and LLaMA, which are powered by these pre-trained FMs, have showcased the vast potential of artificial general intelligence (AGI)\cite{10384606, 10183789, 9606720 }.
These groundbreaking AI systems have the ability to tackle a wide range of complex tasks, including video generation, image content summarization, and continuous dialogue\cite{10398474}. This demonstrates the incredible capabilities that AI has attained, and the critical role it can play in supporting various real-world applications.
By fine-tuning (FT) pre-trained FMs on local datasets\cite{pmlr-v108-radiya-dixit20a}, customized LLMs offer specialized services\cite{10558820}, aligning with the need for establishing native artificial intelligence (AI) in the era of 6G\cite{10558825, tian2024edgecloud, sehad2024generative,10558822}.
Unlike conventional model training, FT updates the network parameters of the pre-trained FMs rather than training from scratch, aiming for enhancing learning performance on local datasets\cite{wu2024netllm}. However, full-model FT for pre-trained FMs (e.g., GPT-4 with 1.76 trillion parameters) still incurs a high communication overhead and computational complexity, hindering their practical deployment.

Parameter-efficient fine-tuning (PEFT) has attracted much attention, given its capability of reducing the number of trainable parameters for pre-trained FMs while achieving comparable learning performance with the full-model FT\cite{Zhu_2023_ICCV, NEURIPS2023_a0054803, vos2022towards, 10.1007/978-3-031-43415-0_31, NEURIPS2022_0cde695b, hu2022lora, bu2024differentially}.
One well-known PEFT method is prompt/prefix tuning\cite{Zhu_2023_ICCV, NEURIPS2023_a0054803, vos2022towards, 10.1007/978-3-031-43415-0_31}, which involves adjusting the added prefix of the embedded tensor. 
Despite its theoretical advantages in few-shot training, the performance of prompt/prefix tuning in practice may not be desirable because of the non-monotonic changes in learning performance with respect to the trainable parameters\cite{hu2022lora}.
In addition, adapter tuning inserts learnable adapters with a few linear layers inside the transformer. 
This method, while achieving the desired performance, unfortunately increases the inference delay due to hindering the parallel processing capabilities.
Furthermore, the bit-fit method\cite{bu2024differentially} fine-tunes the bias term of the pre-trained FMs, leading to much smaller computational overhead but yielding less desired effectiveness compared to other methods.

Meanwhile, by exploiting special properties of the parameters of the pre-trained FMs (e.g., sparsity, low-ranking) and representing the associated adjustments with a few parameters, the reparameterization method reduces computational complexity while preserving parallel processing.
In particular, low-rank adaptation (LoRA) is a popular reparameterization method\cite{hu2022lora}, which represents the adjustment of each matrix in the pre-trained FMs with the multiplication of two low-rank matrices. 
LoRA only updates these added low-rank matrices while freezing the original parameters, where the support for parallel processing of low-rank matrices is retained.
Additionally, LoRA can be easily migrated to various pre-trained FMs.
For instance, the authors in \cite{NEURIPS2023_1feb8787} applied LoRA onto quantized LLM to reduce the storage and training expenses.
The authors in \cite{10472574} applied LoRA to a hierarchical LLM, enhancing the reliability under wide-range of learning tasks.
The authors in \cite{Yu_Chen_Zhou_He_2024} combined LoRA and neuron-indexing technique for efficient editing of LLM.
However, despite the resilience of LoRA under different scenarios, sufficient training data is essential to fine-tune a LoRA-based pre-trained FMs for achieving desired learning performance, which is impractical at a single edge device due to the limited datasets available. 

Federated fine-tuning (FedFT) emerges as a promising framework to refine a global model with a vast amount of distributed datasets in a privacy-preserving manner\cite{10558823, 10558816, 10364357, slora2023, sun2024improving}.
For instance, the authors in \cite{slora2023} explored the impact of data heterogeneity on the learning efficiency of LoRA-based FedFT and proposed a data-driven initialization approach to enhance the robustness of the LoRA-based FedFT and reduce the training cost. 
The authors in \cite{sun2024improving} studied the effect of differential privacy on LoRA-based FedFT and proposed a partially frozen federated LoRA framework to bolster training stability. However, coordinating the edge server and edge devices across wireless networks encounters hurdles due to the stochastic nature of the fading channel and limited radio resources for FedFT.

To facilitate efficient high-dimensional model exchange under limited radio resources, various approaches have been explored in the context of conventional wireless FL, focusing on device scheduling\cite{shi2020joint}, power allocation\cite{10032291, 10355909}, beamforming design\cite{8952884}, and bandwidth.
However, the number of model parameters in wireless FedFT is considerably greater than that in conventional FL, where the aforementioned methods can not be directly applied.
Meanwhile, there are few works proposed in the context of wireless FedFT\cite{lyu2024rethinking, wen2024pretraining}.
In particular, in \cite{lyu2024rethinking}, the authors developed a joint pre-training and fine-tuning framework for FMs, where a joint communication and resource allocation design was proposed to balance trade-off between the learning performance, delay and energy consumption.
Meanwhile, the authors in \cite{wen2024pretraining} theoretically analyzed the convergence bound of pre-training under federated meta learning framework and the generalization error with personalized fine-tuning. 
However, in the aforementioned works, the FMs are assumed to fully deployed at edge devices\cite{10558823, lyu2024rethinking, wen2024pretraining}, which is impractical for FedFT due to limited storage and computational capabilities at the edge devices. 
Additionally, their simulation results were based on the conventional models (e.g., convolution neural network) instead of FMs.
Moreover, the convergence analysis for split wireless FedFT (i.e., partial deployment of models at edge devices), however, has not been studied in the literature yet.

In this paper, we propose a split LoRA-based wireless FedFT framework, where the embedding and task modules are deployed on edge devices, and the computationally-intensive encoder is deployed on the edge server. 
Our objective is to devise an efficient algorithm that jointly optimizes device scheduling and bandwidth allocation to enhance the learning performance of the FedFT system, and address the following challenges.
Firstly, the metric for characterizing the convergence behavior of FedFT in terms of device scheduling and bandwidth allocation is implicit, which impedes dedicated optimizations of radio resources.
Secondly, the learning performance of FedFT is influenced by the number of scheduled edge devices across all communication rounds, necessitating a long-term perspective optimization of device scheduling and bandwidth allocation policies.
Thirdly, the coupling of integer-valued device scheduling and continuous-valued bandwidth allocation leads to NP-hardness and non-convexity.
To tackle these challenges, we begin by analyzing the convergence performance of the system under consideration and formulate an online optimization problem. 
By applying the Lyapunov analysis, we decompose the formulated online problem into a series of sequential sub-problems, which are then solved using the proposed online algorithm.
Our contributions can be summarized as follows.
\begin{enumerate}
	\item [$\bullet$]We develop a split LoRA-based FedFT framework over wireless networks by decomposing the pre-trained FMs into the embedding module, the encoder, and the task module. Specifically, we deploy the embedding and task modules at the edge devices, while keeping the computationally expensive encoder at the edge server. Wireless links sequentially connect the embedding module, encoder, and task module for effective forward inference and backward training. Additionally, by exploiting the low-rank property of the LoRA technique, the proposed framework aggregates the gradient of the task module with respect to the output of a low-rank matrix instead of that of the encoder, thereby reducing communication overhead and enhancing privacy.
	\item [$\bullet$]For the first time, we rigorously analyze the convergence behavior of the wireless FedFT framework across different communication rounds, highlighting the growing importance of increasing the number of scheduled edge devices in the gradient aggregation. 
	Motivated by this analysis, we formulate an online convergence upper bound minimization problem, which requires the joint optimization of device scheduling and bandwidth allocation. 
	\item [$\bullet$]To decouple the impacts of scheduling and bandwidth allocation on the convergence behavior across different communication rounds, we apply Lyapunov analysis to reformulate the online problem as a series of sequential optimization sub-problems. 
	To address the non-convex mixed-integer programming in each sub-problem, we propose a set expansion strategy to separately optimize device scheduling and bandwidth allocation. 
	Furthermore, we examine the structure of these policies and prove the $\Delta$-optimality of our proposed online algorithm.
	\item [$\bullet$] Extensive simulation results presented to validate the effectiveness of the proposed algorithm, demonstrating its performance with various datasets for the LLM and the CIFAR-10 dataset for the LVM. 
	Results show that the proposed joint device scheduling and bandwidth allocation algorithm achieves excellent learning performance in both LLM and LVM scenarios.
\end{enumerate}

The remainder of this paper is organized as follows.
Section \ref{SecSys} presents the learning and signal models of the proposed framework. The convergence analysis and problem formulation are given in Section \ref{SecPro}, followed by the development of an online algorithm in Section \ref{SecSol}. Section \ref{SecSim} presents the simulation results, and Section \ref{SecCon} concludes the paper.
\section{System Model}\label{SecSys}
\subsection{Learning Model}
As shown in Fig. \ref{systemfig}, we consider a LoRA-based wireless FedFT framework, where one single-antenna edge server coordinates a set of $K$ single-antenna edge devices, denoted as $\mathcal{K} = \{1, \ldots, K\}$, to fine-tune a global model based on their local datasets.
\begin{figure}[t]
	\centering
	\includegraphics[width=\linewidth]{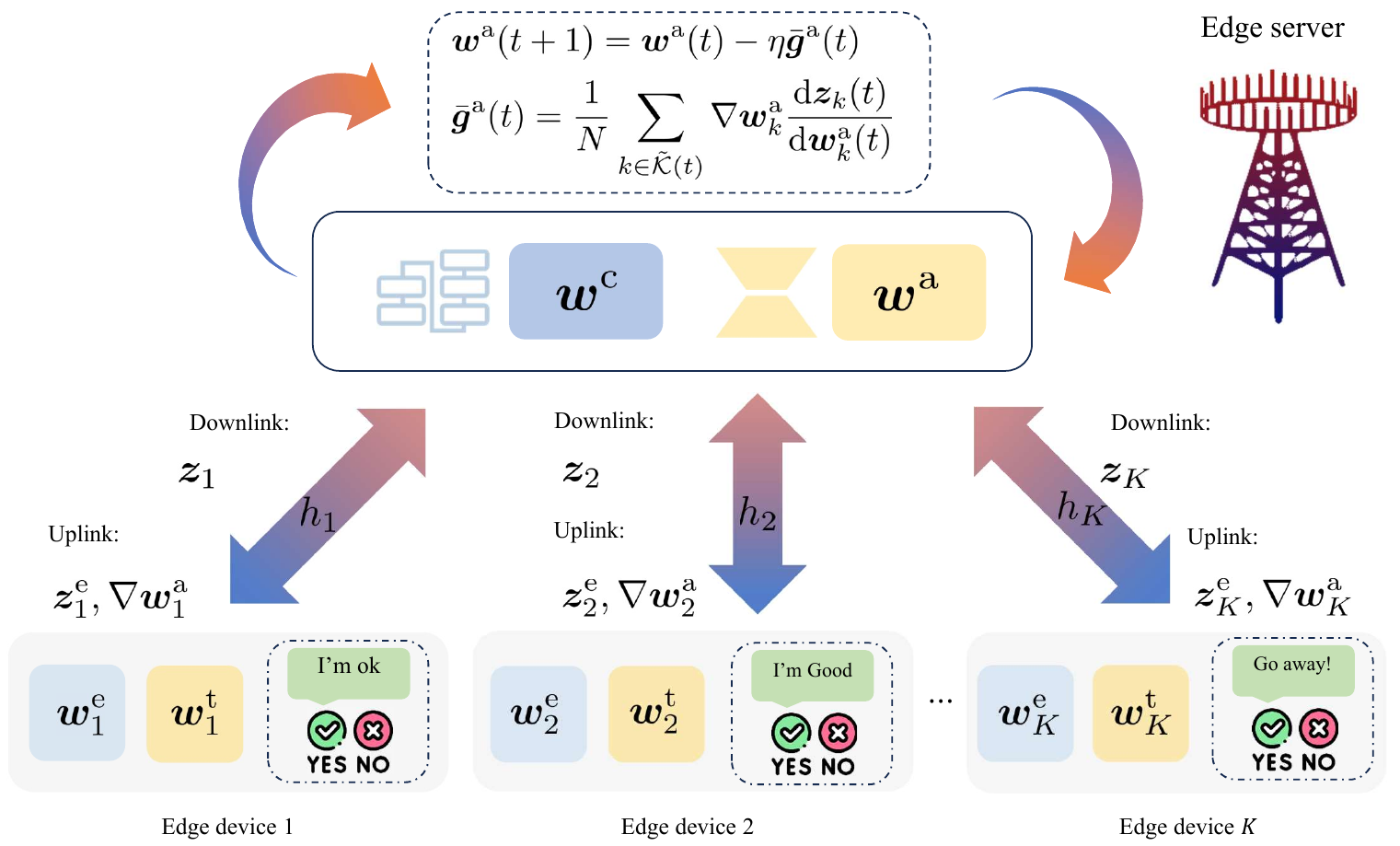}
	\caption{Illustration of the communication process for the proposed FedFT framework.}
	\label{systemfig}
\end{figure}
The local datasets are independent and identically distributed (i.i.d.) and the local dataset of device $k$ consists of $M$ feature-label pairs $\mathcal{D}_k = \{(\bm{x}_{m}, \bm{y}_{m})\}$. 
The objective of FedFT is to find a set of parameters for the task models $\{\bm{w}_{k}^{\rm t}\}_{k\in\mathcal{K}}$ and the low-rank matrices $\bm{w}^{\rm a}$, denoted by $\bm{W} = \{\bm{w}^{\rm a}\}\cup\{\bm{w}^{\rm t}_{k}\;|\;\forall k\in\mathcal{K}\}$ based on the off-the-shelf pre-trained FM $\bm{w}^{\rm f}$ that minimizes the global loss function $F(\bm{W}; \bm{w}^{\rm f}, \{\mathcal{D}_k\})$, i.e.,
\begin{align}\label{EFT_obj}
	\begin{split}
		\bm{W}^{\star} = &\underset{\{\bm{w}^{\rm t}_{k}\}, \bm{w}^{\rm a}}{\arg\min }\;F\left(\bm{W};\bm{w}^{\rm f},\{\mathcal{D}_k\}\right)\\
		=&\frac{1}{K}\sum_{k=1}^{K} f_k\left(\bm{w}^{\rm t}_{k}, \bm{w}^{\rm a};\bm{w}^{\rm f}, \{\mathcal{D}_k\}\right),
	\end{split}
\end{align}
where $f_k(\cdot)$ denotes the local loss function of edge device $k\in\mathcal{K}$.

By splitting the pre-trained FM $\bm{w}^{\rm f}$ into embedding module $\bm{w}^{\rm e}$, encoder $\bm{w}^{\rm c}$, and task module $\bm{w}^{\rm t}$, i.e., $\bm{w}^{\rm f} = \{\bm{w}^{\rm e}, \bm{w}^{\rm c}, \bm{w}^{\rm t}\}$, we deploy the embedding and task modules at the edge devices, while the encoder is most computationally expensive and remains on the edge server.
Additionally, we apply LoRA to the encoder for FedFT, where the original parameters are frozen and a set of low-rank matrices are added to be collaboratively fine-tuned.
Consequently, the forward inference of the proposed wireless FedFT system is performed as follows. 
The local data are encoded by the local embedding module, and then transmitted to the edge server for further feature extraction in the encoder. Subsequently, the corresponding task module post-processes the features for specific task outputs. 
Fig. \ref{process_flow} shows the overall workflow of the considered FedFT system, with detailed forward inference and backward propagation provided in Fig. \ref{forward_inf} and Fig. \ref{backward_Pro}, respectively.
We apply the federated averaging (FedAvg) algorithm to the low-rank matrices. 
Specifically, in each learning epoch, the following steps are sequentially performed.
\begin{figure}[t]
	\centering
	\subfigure[Forward inference]
	{
		\begin{minipage}[t]{\linewidth}
			\centering
			\includegraphics[width=\linewidth]{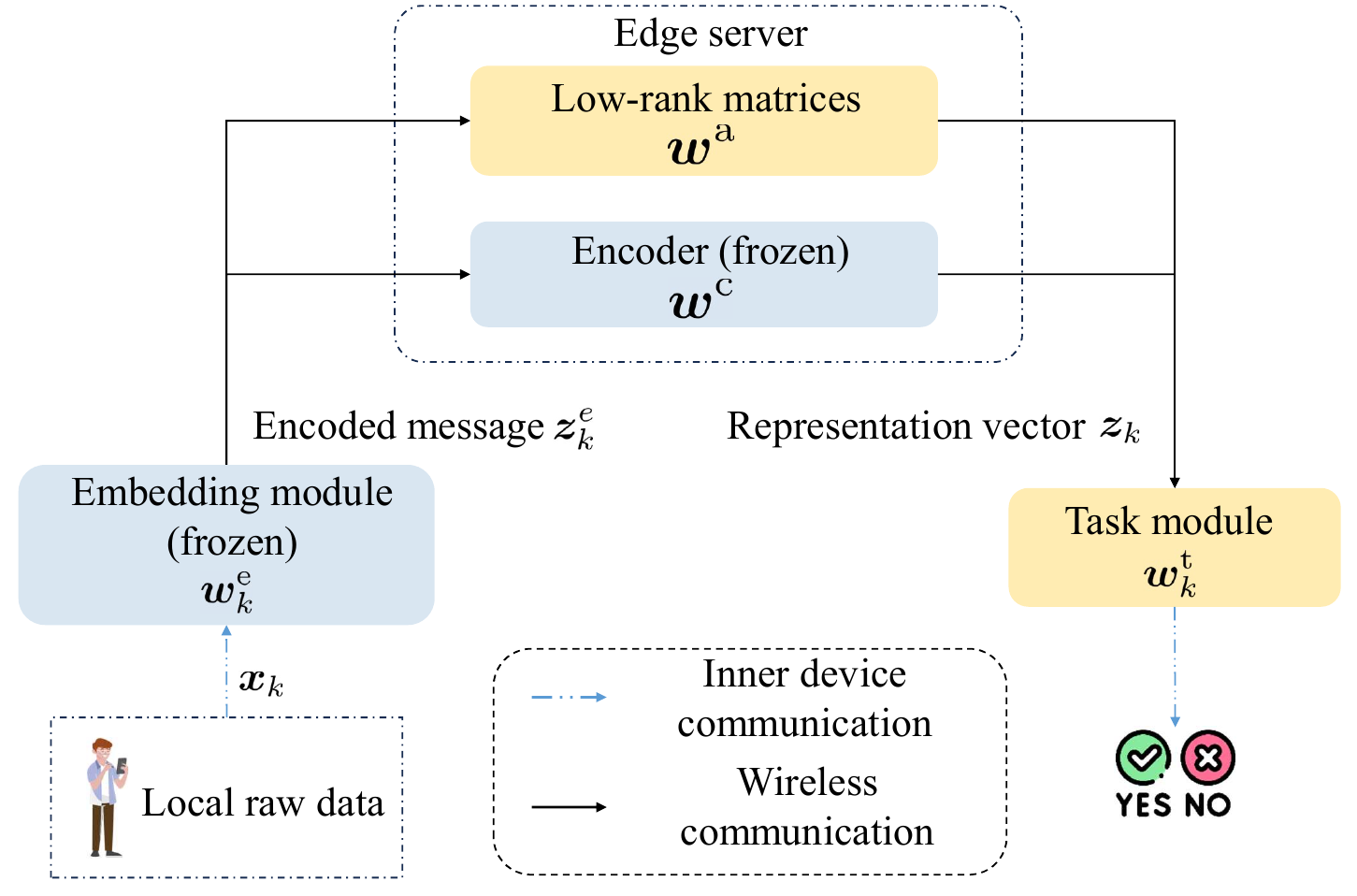}
			\label{forward_inf}
		\end{minipage}
	}
	\\
	\subfigure[Backward propagation]
	{
		\begin{minipage}[t]{\linewidth}
			\centering 
			\includegraphics[width=\linewidth]{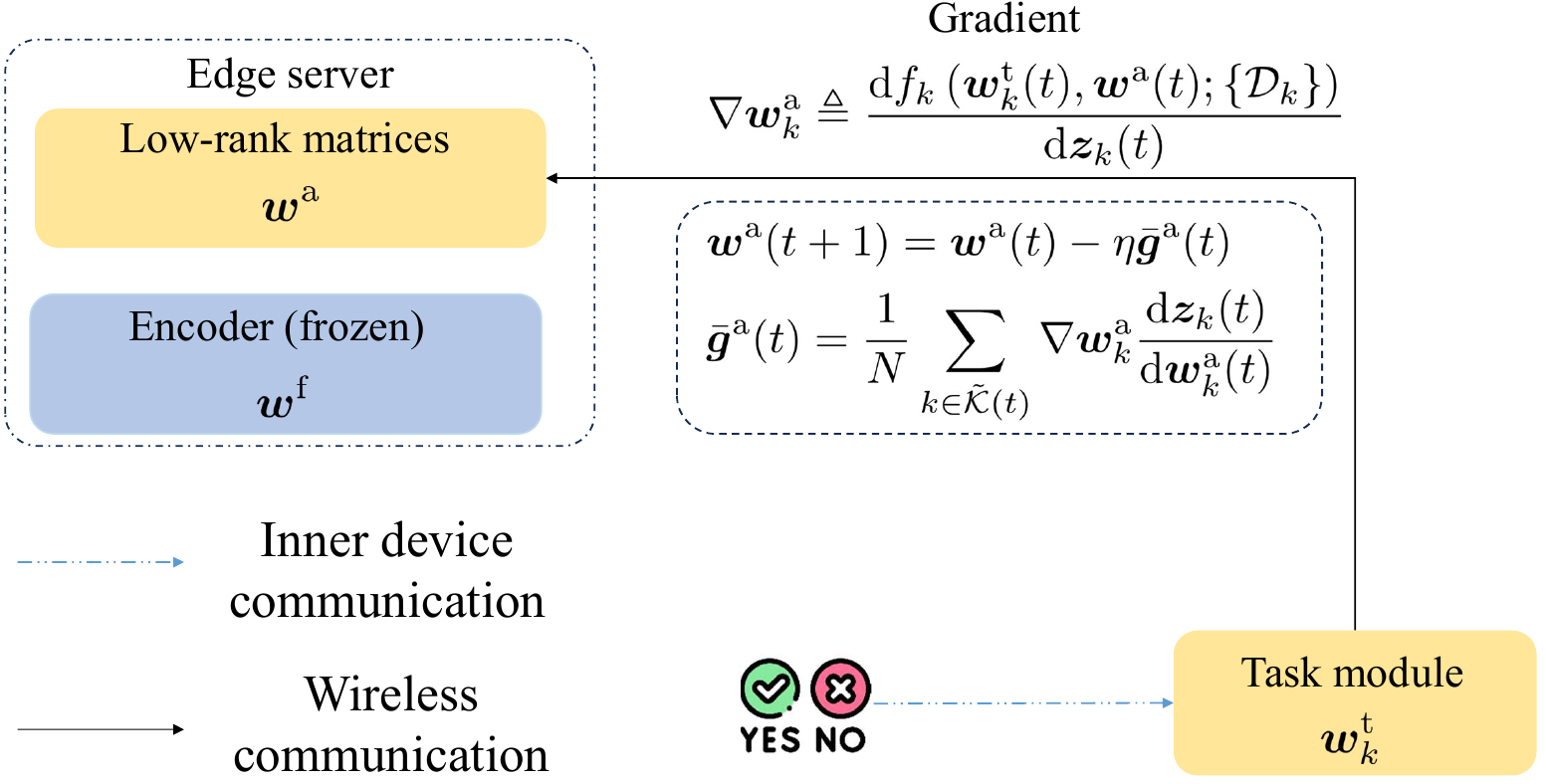}
			\label{backward_Pro}
		\end{minipage}
	}
	\caption{Illustration of forward inference and backward propagation for the considered FedFT system.}
	\label{process_flow}
\end{figure}
\begin{enumerate}
	\item[$\bullet$] \textit{Edge Device Scheduling}: 
	Due to the limited radio resources, the edge server schedules $N$ of $K$ edge devices, denoted by a subset $\tilde{\mathcal{K}}(t)$, to participate in the FedFT of the current round.
	\item[$\bullet$] \textit{Global Feature Representation}: Each scheduled edge device encodes its local feature through $\bm{w}_k^{\rm{e}}$ as $\bm{z}^{\rm e}_k(t)$, which is then transmitted to the edge server. Upon reception, the edge server further calculates the associated feature representation vectors of each encoded message $\{\bm{z}^{\rm e}_k(t)\}$ according to $\bm{w}^{\rm{c}}$ and $\bm{w}^{\rm a}(t)$, denoted as $\{\bm{z}_k(t)\}$, and then transmits them to the corresponding edge devices for local FT.
	\item[$\bullet$] \textit{Local Model Aggregation}: Based on the received feature representation vectors ${\bm{z}_k(t)}$ and local labels, each scheduled edge device computes the gradient of the task module $\bm{w}^{\rm t}_k(t)$ and gradient information of the low-rank matrices $\bm{w}^{\rm{a}}$, denoted as $\bm{g}_k^{\rm t}(t)$ and
	\begin{equation}
		{\nabla} \bm{w}^{\rm{a}}_{k} \triangleq\frac{{\rm{d}}f_k\left(\bm{w}_k^{\rm t}(t), \bm{w}^{\rm a}(t); \{\mathcal{D}_k\}\right)}{{\rm{d}}\bm{z}_k(t)},
	\end{equation} respectively. 
	The task module of each edge device is tailored for a specific task, which may differ from other task models. Meanwhile, the low-rank matrices in the pre-trained FM are designed for general use.
	Therefore, each scheduled edge device updates its own task module as follows
	\begin{align}
		\bm{w}_k^{\rm t}(t+1) &= \bm{w}_k^{\rm t}(t) - \eta \bm{g}_k^{\rm t}(t),
	\end{align}
	and uploads ${\nabla} \bm{w}^{\rm{a}}_{k}$ onto the edge server for performing the FedAvg algorithm. 
	\item[$\bullet$] \textit{Global Model Update}: Following the chain rule, the edge server calculates the averaged gradient of the low-rank matrices as follows
	\begin{equation*}
		\bar{\bm{g}}^{\rm a}(t) = \frac{1}{N}\sum_{k\in\tilde{\mathcal{K}}(t)} {\nabla} \bm{w}^{\rm{a}}_{k} \frac{{\rm{d}}\bm{z}_k(t)}{{\rm{d}}\bm{w}^{\rm{a}}_k(t)},
	\end{equation*}
	which is used for updating the global low-rank matrices
	\begin{align}
		\bm{w}^{\rm a}(t+1) &= \bm{w}^{\rm a}(t) - \eta \bar{\bm{g}}^{\rm a}(t),
	\end{align}
	where $\eta$ denotes the learning rate. 
\end{enumerate}

\subsection{Signal Model}
We consider a frequency division multiple access (FDMA) system.
By denoting $h_{k}(t)$ as the block-fading channel coefficient between the edge server and edge device $k\in\mathcal{K}$, the achievable transmission rate in the uplink can be given by 
\begin{equation}
	r_k(t) = B_k(t) \log_2\left(1+\frac{|h_k(t)|^2}{B_k(t)\sigma^2}\right),\forall\;k\in\tilde{\mathcal{K}}(t),
\end{equation}
where $B_k(t)\geq 0$ denotes the allocated bandwidth for edge device $k\in\tilde{\mathcal{K}}(t)$, and $\sigma^2$ denotes the power spectral density of additive white Gaussian noise (AWGN). 
In particular, a total bandwidth $B$ is shared by all the scheduled edge devices, i.e.,
\begin{equation}\label{sum_bandwidth}
	B = \sum_{k\in\tilde{\mathcal{K}}(t)}B_k(t).
\end{equation}
The transmission delay among all scheduled edge devices is given by
\begin{equation}
	D(t) =\underset{k\in\tilde{\mathcal{K}}(t)}{\max}\;\frac{\mu}{r_k(t)}, 
\end{equation}
where $\mu$ denotes the total length in bits of the transmitted symbols for each scheduled edge device.
Besides, we consider an average transmission delay constraint as in \cite{9606731}, i.e.,
\begin{equation}\label{averdelay}
	\frac{1}{T}\sum_{t=1}^{T}D(t) \leq \bar{D}.
\end{equation}
It is worth noting that the transmitted symbols include the embedding message, the associated feature representation vectors, and the backward gradient information. Additionally, $\mu$ is jointly determined by the size of the training batch and the rank of added matrices, where a larger batch size results in higher transmission overhead.

\section{Problem Formulation}\label{SecPro}
In this section, we derive the upper bound of the optimality gap between the global loss under an arbitrary scheduling policy and the ideal case for split LoRA-based wireless FedFT system, and formulate an upper bound minimization problem.
\subsection{Assumptions and Convergence Results}
To begin with, we make several widely-adopted assumptions to facilitate the theoretical analysis.
\begin{assumption}\label{ass-1}
    There exists a set of parameters $\bm{W}^{\star}$ that achieves the global minimum of the global loss function
    \begin{equation}
        F(\bm{W}^{\star})\leq F(\bm{W}), \forall\; \bm{W}.
    \end{equation}
\end{assumption}
\begin{assumption}\label{ass-2}
The loss function $f_i(\cdot)$ is non-convex and $L$-smooth, i.e., 
\begin{align}
    \Vert\nabla f_{i}(\bm{x})-\nabla f_{i} (\bm{x}^{\prime})\Vert_2\leq L\Vert\bm{x}-\bm{x}^{\prime}\Vert_2, \; L>0.
\end{align}
\end{assumption}
\begin{assumption}\label{ass-3}
    The local gradient $\bm{g}^{\rm a}_k(t)$ is an unbiased estimate of $\nabla f_k(\tilde{\bm{w}}_{k}(t))$, i.e.,
    \begin{align}
        \mathbb{E}\left[\bm{g}^{\rm a}_k(t)\right] &= \nabla f_k(\tilde{\bm{w}}_{k}(t)).
    \end{align}
\end{assumption}
\begin{assumption}\label{ass-4}
The variance of all entries of $\bm{g}^{\rm a}_k(t)$ and $\bm{g}_{k}^{\rm t}(t)$ are upper bounded by a constant $\phi^2\geq 0$.
\end{assumption}
\begin{assumption}\label{ass-5}
	There exists a constant $\tau\geq 0$ such that the Polyak-Łojasiewicz inequality holds for $F(\bm{W})$, i.e., 
	\begin{equation*}
		\frac{1}{2}\Vert\nabla F(\bm{W})\Vert^2_2 \geq \tau \left(F(\bm{W}) - F(\bm{W}^*)\right).
	\end{equation*}
\end{assumption}
\begin{lemma}\label{lmm1}
 Given the objective in \eqref{EFT_obj}, the convergence behavior of $F(\bm{W};\{\mathcal{D}_k\})$ can be separately optimized with respect to the task module $\bm{w}^{\rm t}$ and the added low-rank matrices $\{\bm{w}^{\rm a}_k\}$,
    \begin{equation}\label{EFFT_obj_sep} 
        \begin{split}
        \left\Vert\nabla F(\bm{W})\right\Vert^2\leq
        \frac{2}{K^2}&\left(\underset{\text{Task module}}{\underbrace{\sum_{k\in\mathcal{K}}\left\Vert\nabla_{\bm{w}^{\rm t}}f_k(\bm{w}^{\rm t};\bm{w}_k^{\rm a})\right\Vert^2}}\right.\\
        &\left.+\underset{\text{Added low-rank matrices}}{\underbrace{\sum_{k\in\mathcal{K}}\!\left\Vert\nabla_{\bm{w}^{\rm a}_{k}}f_k(\bm{w}^{\rm a}_k; \bm{w}^{\rm t})\right\Vert^2}}\right).
        \end{split}
    \end{equation}
\end{lemma}
\begin{proof}
    See Appendix A.
\end{proof}
According to Lemma \ref{lmm1} and the Assumptions \ref{ass-1}-\ref{ass-5}, we have the following theorem.
\begin{thm}\label{thm1}
    Given an arbitrary device scheduling policy, when $L<\frac{\eta }{\eta^2+1}$,
   the optimality gap after $T+1$ communication rounds is upper bounded by 	\begin{equation}\label{Opt_Gap}
	    \begin{split}
	        &F(\bm{W}(T+1))- F(\bm{W}^{\star})\\
	        \leq&\left(1-2\tau \varsigma(T)\right)(F(\bm{W}(T)) - F(\bm{W}^\star))+\beta-\alpha(T)\\
	        =&(1-2\tau \varsigma(T))\left[F(\bm{W}(T)) - F(\bm{W}(T-1)) \right.\\
	        &\left.+ F(\bm{W}(T-1)) -   F(\bm{W}^\star)\right]+\beta-\alpha(T)\\
	        \leq&\underset{\text{Initial gap}}{\underbrace{\left(\prod_{i=0}^{T}(1-2\tau \varsigma(i))\right)F(\bm{W}(0)) -F(\bm{W}^\star)}} \\
	    &-\underset{\text{LoRA related gap}}{\underbrace{ \sum_{i=1}^{T}\left(\prod_{j=0}^{i-1}\left(1-2\tau \varsigma(T-j)\right) \right)\alpha(T-i)-\alpha(T)}}\\
	    &+\underset{\text{Task module related gap}}{\underbrace{\beta\left(1+\sum_{i=1}^{T}\left(\prod_{j=0}^{i-1}\left(1-2\tau \varsigma(T-j)\right) \right)\right)}}
	    \end{split}
	\end{equation}
    where 
        \begin{align}
            &\varsigma(t)=-\frac{L\eta^2-\eta}{K^2}-\frac{(K-{N(t)})L}{2{N(t)}(K-1)K^2},\label{varisigma_conv}\\
            &\alpha(t)=\phi^2K^2\varsigma(t)\Omega^{\rm a}, \label{alpha_conv}\\
            &\beta = \Omega^{\rm t}\left(L\eta^2 
            -\eta\right)\phi^2 \label{beta_conv}.
        \end{align}
        $\Omega^{\rm a}$ and $\Omega^{\rm t}$ denote the number of elements in $\bm{g}^{\rm a}_k(t)$ and $\bm{g}_{k}^{\rm t}(t)$, respectively.
\end{thm}
According to Theorem \ref{thm1}, we have the following observations:
\begin{enumerate}
	\item[$\bullet$]  \textbf{Diminishing initial gap}: 
	When 
	\begin{equation*}
		L<\frac{\eta {N(t)}(K-1)}{(K-1){N(t)}\eta^2+K-{N(t)}}<\frac{\eta }{\eta^2+1},
	\end{equation*}
	the positive weight coefficient  $\prod_{i=0}^{T}(1-2\tau \varsigma(i))$ is exponentially decaying with respect to the number of communication rounds. As a result, the impact of the initial gap vanishes as $T\to \infty$.
	
	\item [$\bullet$] \textbf{Separability of the optimality gap}: The optimality gap can be divided into the initial gap, the LoRA-related gap, and the task module-related gap. Despite the initial gap diminishes as the number of communication rounds increases, the optimality gap is jointly determined by the Lipschitz constant $L$, learning rate $\eta$, gradient size of the task module $\Omega^{\rm t}$, and the added low-rank matrices $\Omega^{\rm a}$, the number of scheduled edge devices $N$ and the total number of edge devices $K$, variance of the local gradient $\phi^2$, and number of communication rounds $T$. Enlarging the mini-batch size in local training reduces the value of $\phi^2$, thereby enhancing convergence performance.
	Furthermore, deploying a smaller task model reduces the task module-related gap by decreasing $\beta$, thus improving the convergence performance of the FedFT system. Once these parameters are fixed, the optimality gap is dominated by the weighted accumulated terms, including the LoRA-related gap and the task module-related gap, where the weighted coefficient $\prod_{j=0}^{i-1}\left(1-2\tau \varsigma(T-j)\right)$ depends on the number of scheduled edge devices. This motivates us to optimize scheduling and resource allocation policies for enhanced convergence performance.
	\item [$\bullet$] \textbf{Importance of increasing the number of scheduled edge devices}: 
	As observed in \eqref{varisigma_conv}, $\varsigma(t)$ is a monotonically increasing function with respect to $N(t)$. Thus, $\prod_{j=0}^{i-1}\left(1-2\tau \varsigma(T-j)\right)$ can be minimized by optimizing $\{N(t)\}_{t=1}^{T}$. Specifically, given the learning rate $\eta$ and Lipschitz constant $\tau$, enlarging $N(t)$ reduces the value of $1-2\tau \varsigma(t)$, where $N(t)$ is jointly determined by scheduling and resource allocation policies.
	Moreover, the weighted coefficients $\prod_{j=0}^{i-1}\left(1-2\tau \varsigma(T-j)\right)$ in later communication rounds are more dominant than those in initial rounds. Therefore, the convergence performance of the FedFT system should be enhanced by optimizing scheduling and bandwidth allocation policies from a long-term perspective.
\end{enumerate}

Built upon the above observations, we shall develop an efficient online optimization algorithm that minimizes the optimality gap from a long-term perspective. By eliminating the diminishing and constant terms in \eqref{Opt_Gap}, the corresponding optimization problem can be formulated as follows
\begin{subequations}\label{Pro_0}
    \begin{align}
       \underset{\tilde{\mathcal{K}}(t), \{B_k(t)\}}{\max}\;&\sum_{i=1}^{T}\left(\prod_{j=0}^{i-1}\left(1-2\tau \varsigma(T-j)\right) \right)\alpha(T-i)-\alpha(T)\\
        \st{}\;\quad& \eqref{sum_bandwidth},\eqref{averdelay},\\
        & 0\leq B_k(t)\leq B,\;\forall\;k=1,\ldots,N(t).
    \end{align}
\end{subequations}
Resolving Problem \eqref{Pro_0} is challenging for the following reasons. 
First, due to the average transmission delay constraint and the accumulated weighted LoRA-related gaps in the objective function,
Problem \eqref{Pro_0} shall be resolved from a long-term perspective.
Second, the coupling of $\tilde{\mathcal{K}}(t)$ and $\{B_k(t)\}$ results in the non-convexity of Problem \eqref{Pro_0}.
To address these challenges, we propose transforming Problem \eqref{Pro_0} into a series of online problems using Lyapunov analysis, and then develop an online algorithm for the joint optimization of bandwidth allocation $\{B_k(t)\}$ and device scheduling $\tilde{\mathcal{K}}(t)$, maximizing the convergence performance of the FedFT system.

\section{Proposed Method}\label{SecSol}
In this section, we transform the long-term optimization Problem \eqref{Pro_0} into a series of deterministic sub-problems via Lyapunov analysis. Then, we develop an efficient online algorithm for joint optimization of bandwidth allocation and device scheduling.

\subsection{Problem Decomposition via Lyapunov Analysis}\label{Sec_Lya_ana}
To address the time average constraint \eqref{averdelay}, we transform \eqref{averdelay} into a virtual queue based constant via the Lyapunov analysis\cite{9152999, 9449944}, which enables us to represent the long-term dynamics of transmission delay with the following virtual queue based update equation
\begin{subequations}\label{vir_que}
    \begin{align}
        \hat{D}(t) &= \max\left\{0, \hat{D}(t-1)+D(t)-\bar{D}(t) \right\},\\
        \hat{D}(0) &= 0,
    \end{align}
\end{subequations} 
where $\hat{D}(t)$ denotes the virtual queue with respect to the transmission delay $D(t)$.
Note that pursuing the asymptotic stabilization of $\hat{D}(t)$ ensures the satisfaction of \eqref{averdelay}, where $\hat{D}(t)$ fluctuates due to the randomness of wireless channel.
Thus, we introduce the Lyapunov function $V\left(\hat{D}(t)\right)$ and Lyapunov drift $ \Delta V\left(\hat{D}(t)\right)$ to monitor the fluctuation of $\hat{D}(t)$ as follows
\begin{align}
    V\left(\hat{D}(t)\right) =& \frac{1}{2}\hat{D}^2(t),\label{lya_fun}\\
    \Delta V\left(\hat{D}(t)\right) = &\mathbb{E}\left[V\left(\hat{D}(t+1)\right) - V\left(\hat{D}(t)\right)\;|\; \hat{D}(t)\right], \label{lya_drift}
\end{align}
where $V\left(\hat{D}(t)\right)$ is a quadratic function with respect to $\hat{D}(t)$ representing the accumulated delay, and $\Delta V\left(\hat{D}(t)\right)$ denotes the changes between the current and next communication round. 
Minimizing \eqref{lya_drift} leads to the stabilization of \eqref{vir_que}, but $V\left(\hat{D}(t)\right)$ is unknown to \eqref{lya_drift} at the $t$-th communication round, which hinders the resource allocation and selection for size of gradient information at each communication round.
To circumvent this problem, we approximate \eqref{lya_drift} with the following lemma
\begin{lemma}\label{lmm2}
    The upper bound of \eqref{lya_drift} at the $t$-th communication round can be expressed by
    \begin{equation}
        \Delta V\left(\hat{D}(t)\right)\leq \hat{D}(t)\left(D(t)-\bar{D}(t)\right).
    \end{equation}
\end{lemma}
\begin{proof}
    See Appendix C.
\end{proof}

Building upon Lemma \ref{lmm2}, we rewrite Problem \eqref{Pro_0} by applying drift-and-penalty \cite{9152999} as follows
    \begin{align}\label{Pro_1}
        \textbf{P1}\underset{\tilde{\mathcal{K}}(t), \{B_k(t)\}}{\max}\;&\mathcal{J}( \tilde{\mathcal{K}}(t), \{B_k(t)\})\triangleq N(t) - \zeta(t)\hat{D}(t)D(t)\notag\\
        \st{}\quad\quad& \eqref{sum_bandwidth}, 0\leq B_k(t)\leq B,\;\forall\;k=1,\ldots,N,
    \end{align}
where $\zeta(t)$ denotes a descending positive control parameter that adaptively handles the tradeoff between maximizing the number of edge devices participating in FedFT and minimizing the transmission delay of all scheduled edge devices.
It is noteworthy that the penalty term (i.e., $\zeta(t)\hat{D}(t)D(t)$) increases when the transmission delay at communication round $t$ is larger than the average transmission delay, which in turn increases the weight of the penalty term in the future communication rounds and verse versa.
Note that Problem \textbf{P1} is a NP-hard mixed-integer programming (MIP) problem due to the discrete-valued device scheduling and size of gradient information and continuous-valued bandwidth allocation.

\subsection{Online Optimization}\label{Sec_Online_opt}
To maximize the number of scheduled edge devices while stabilizing the virtual queue from a long-term perspective, we develop an online algorithm to jointly optimize the scheduling set and bandwidth at each communication round.
We adopt the set expansion strategy and propose to gradually add the edge device with strong channel condition into the scheduling set $\tilde{\cal K}(t)$.
Specifically, we sort the edge devices in a descending order according to channel power gain $|h_k(t)|^2$ and initialize the scheduling set with $\check{\cal K}(t) = \varnothing$.
For the simplicity of following analysis, we make the following assumption without loss of generality
\begin{equation}\label{reorder_channel_gain}
    |h_1(t)|^2\geq \cdots\geq|h_K(t)|^2.
\end{equation}
Then, the scheduling set $\check{\cal K}(t)$ appends the edge devices one by one based on the order in \eqref{reorder_channel_gain}.
Given an arbitrary scheduling set $\check{\cal K}(t)$, the bandwidth allocation optimization problem can be expressed as
\begin{equation}\label{Pro_2}
    \begin{split}
        \textbf{P2}\quad\underset{\{B_k(t)\}}{\min}\;\quad& \hat{D}(t)D(t)\notag\\
            \st{}\quad\quad& \eqref{sum_bandwidth}, 0\leq B_k(t)\leq B,\;\forall\;k=1,\ldots,N.
    \end{split}
\end{equation}
Recall that $\hat{D}(t)\geq 0$ and $D(t) =\underset{k\in\tilde{\mathcal{K}}(t)}{\max}\;\frac{\mu}{r_k(t)}$, Problem \textbf{P2} can be rewritten by 
\begin{equation}\label{Pro_3}
    \begin{split}
        \textbf{P3}\quad\quad\text{find}\;\quad& \{B_k(t)\}\notag\\
            \st{}\;\quad& \frac{\mu}{r_k(t)} = \check{D}(t),\;\forall\;k\in\check{\cal K}(t)\\
            &\eqref{sum_bandwidth}, 0\leq B_k(t)\leq B,\;\forall\;k\in\check{\cal K}(t),
    \end{split}
\end{equation}
where $\check{D}(t) = \underset{k\in\check{\mathcal{K}}(t)}{\min\max}\;\frac{\mu}{r_k(t)}$ denotes the minimum transmission delay under given scheduling set. 
To this end, given $\check{D}(t)$, we tackle Problem \textbf{P3} by solving the following equation
\begin{equation}\label{Pro_4}
    \frac{\mu}{B_k(t) \log_2(1+\frac{|h_k(t)|^2}{B_k(t)\sigma^2})} = \check{D}(t).
\end{equation}
\begin{lemma}\label{lmm3}
    For any $k\in\check{\cal K}(t)$, given transmission delay $\check{D}(t)$, the required bandwidth for \eqref{Pro_4} is 
    \begin{equation}\label{req_band}
        \frac{1}{B^\star_k} = -\frac{\check{D}(t){\rm LambW}\left(\nu(t)2^{\nu(t)}\ln{2} \right)}{\mu\ln 2}- \frac{\sigma^2}{|h_k(t)|^2},\forall\;k\in\check{\cal K}(t),
    \end{equation}
    where $\nu(t) = -\frac{\mu|h_k(t)|^2}{\check{D}(t)\sigma^2}$, and ${\rm LambW}(\cdot)$ is a Lambert-W function.
\end{lemma}
\begin{proof}
    See Appendix D.
\end{proof}

Based on Lemma \ref{lmm3}, we can obtain the minimum required bandwidth to reach the required transmission delay $\check{D}(t)$, which however may violate constraint \eqref{sum_bandwidth}. 
We can find the optimal $\check{D}(t)$ by a bisection search.
Therefore, Problem \textbf{P2} is finally solved with arbitrary scheduling set $\check{\cal K}(t)$.
Next, we search over $\tilde{\cal K}(t)$ to maximize $\mathcal{J}( \tilde{\mathcal{K}}(t), \{B_k(t)\})$ by adding edge devices one-by-one into $\check{\cal K}(t)$. 
In particular, when $N(t)$ gets larger, each scheduled edge device can share less bandwidth, therefore the corresponding transmission delay increases. 
As a result, a terminate condition can be set as when $\mathcal{J}( \tilde{\mathcal{K}}(t), \{B_k(t)\})$ reduces after including one more edge device.
To this end, we summarize the proposed online algorithm in Algorithm \ref{online_algorithm}.
\begin{algorithm}[!t]
	\caption{Proposed Online Algorithm}
	\LinesNumbered
	\label{online_algorithm}
    \KwIn{$\hat{D}(0)\gets 0$}
	\For{each communication round $t = 1,\ldots,$}{
        $\tilde{\cal{K}}(t) = \{\}$, $\rm obj(t) \gets 0$\\
    Sort the edge devices according to $|h_k(t)|^2$ in descending order, i.e., $|h_1(t)|^2 \geq \cdots\geq |h_K(t)|^2$\\
    \For {$k = |\tilde{\cal{K}}(t)|+1,\ldots, K$}{
        $\check{\cal{K}}(t)\gets \tilde{\cal{K}}(t)\cup \{k\}$\\
        \While{$\sum_{k\in\check{\cal K}(t)}B_k(t)<B$}{
        Bisection search for minimum delay $\check{D}(t)$\\
        Obtain $\{B_k(t)\}$ according to \eqref{req_band}\\
        }
        \uIf{ $\mathcal{J}\left(\check{\cal{K}}(t), \{B_k(t)\}\right)>=\rm obj$}{
            ${\rm obj} \gets \mathcal{J}\left(\check{\cal{K}}(t), \{B_k(t)\}\right)$\\
            $\tilde{\cal{K}}(t) \gets \check{\cal{K}}(t)$\\
            $\{B_k^\star(t)\} \gets \{B_k(t)\}$
        }
        \Else{
            Stop Iteration\\
            \Return{$\{B_k^\star(t)\}$, $\check{D}(t)$, and $\tilde{\cal{K}}(t) $}
        }
    }
    $\hat{D}(t) \gets \max\left\{0, \hat{D}(t-1)+\zeta(t)\left(\check{D}(t)-\bar{D}(t)\right) \right\}$\\
		}
\end{algorithm}
\begin{definition}
    Existing a positive constant $\Delta$ and an optimal solution $\{\tilde{\mathcal{K}}^\star(t), \{B_k^\star(t)\}\}$ for Problem \textbf{P1}, if
    \begin{equation}
    \mathcal{J}( \tilde{\mathcal{K}}(t), \{B_k(t)\})\geq \mathcal{J}( \tilde{\mathcal{K}}^\star(t), \{B_k^\star(t)\}) - \Delta,
    \end{equation}
    then $\{\tilde{\mathcal{K}}(t), \{B_k(t)\}\}$ is $\Delta-$optimal for Problem \textbf{P1}. 
\end{definition}
\begin{thm}\label{thm2} (Threshold-Based $\Delta-$optimal device scheduling)
	The proposed online algorithm achieves the $\Delta-$optimal for solving Problem \textbf{P1}, where the device scheduling has a threshold-based structure given by
	\begin{equation}
			{\rm \rho}_j = \left\{\begin{array}{rc}
				1&j\leq k\\
				0&j>k
			\end{array}.\right.
	\end{equation}
\end{thm}
\begin{proof}
	See Appendix E.
\end{proof}
\section{Simulation Results}\label{SecSim}
In this section we validate the theoretical analysis and evaluate the performance of the proposed algorithm. 
\subsection{Setting}
We consider a wireless FedFT system. 
The edge server is located at $(0, 0)$, and the edge devices are uniformly located in a circle with center $(300, 0)$ and radius $50$ meters. In particular, the large scale fading is modeled as $(d/d_0)^{-3.5}$, where $d_0 = 10$ denotes the reference distance. The noise spectral density is set as $10^{-11}$ W and the total bandwidth $B = 10$ MHz. The average transmission delay $\bar{T} = 0.05$ ms per bit.
\begin{figure*}[t]
	\centering
	\subfigure[Training loss (SST-2)]
	{
		\begin{minipage}[t]{0.31\textwidth}
			\centering
			\includegraphics[width=1\linewidth]{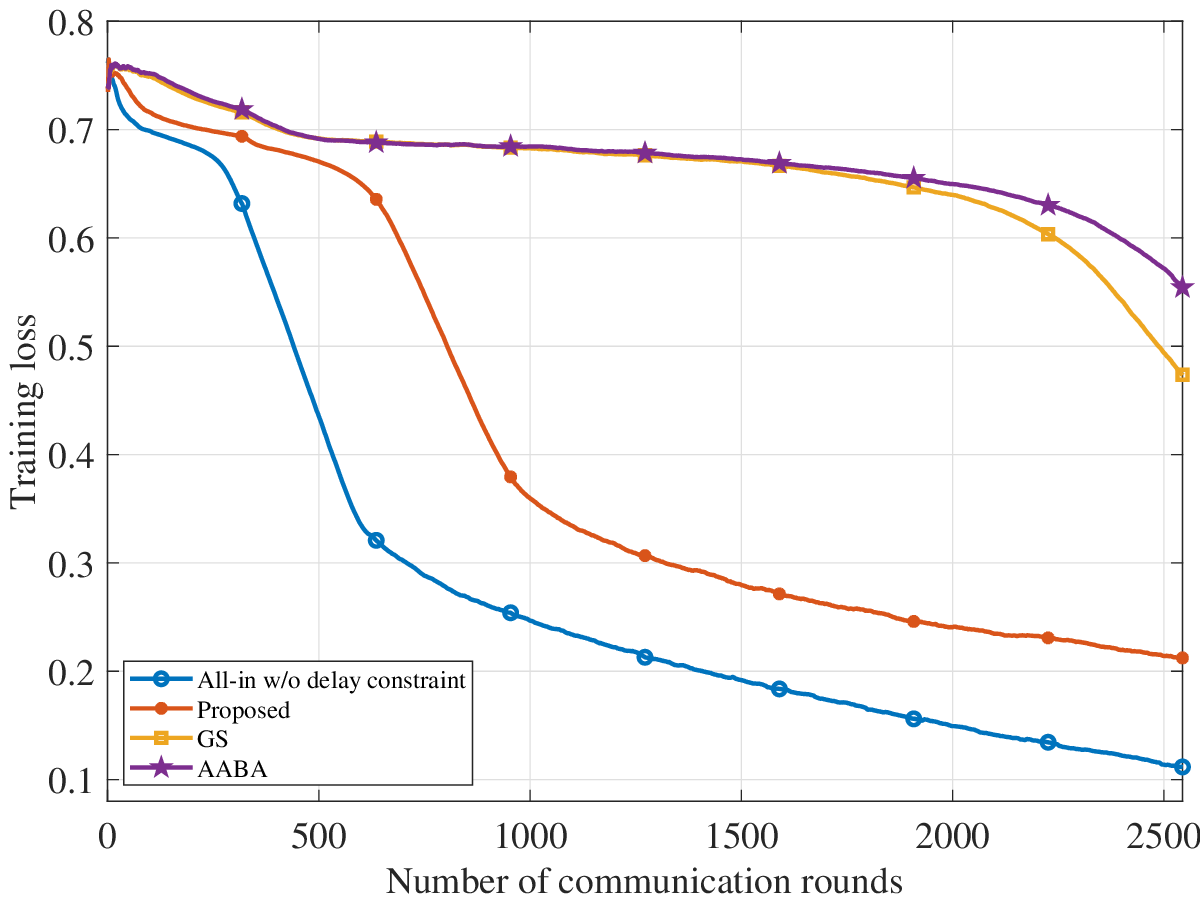}
			\label{fig1_1}
		\end{minipage}
	}
	\subfigure[Training loss (MRPC)]
	{
		\begin{minipage}[t]{0.31\textwidth}
			\centering
			\includegraphics[width=1\linewidth]{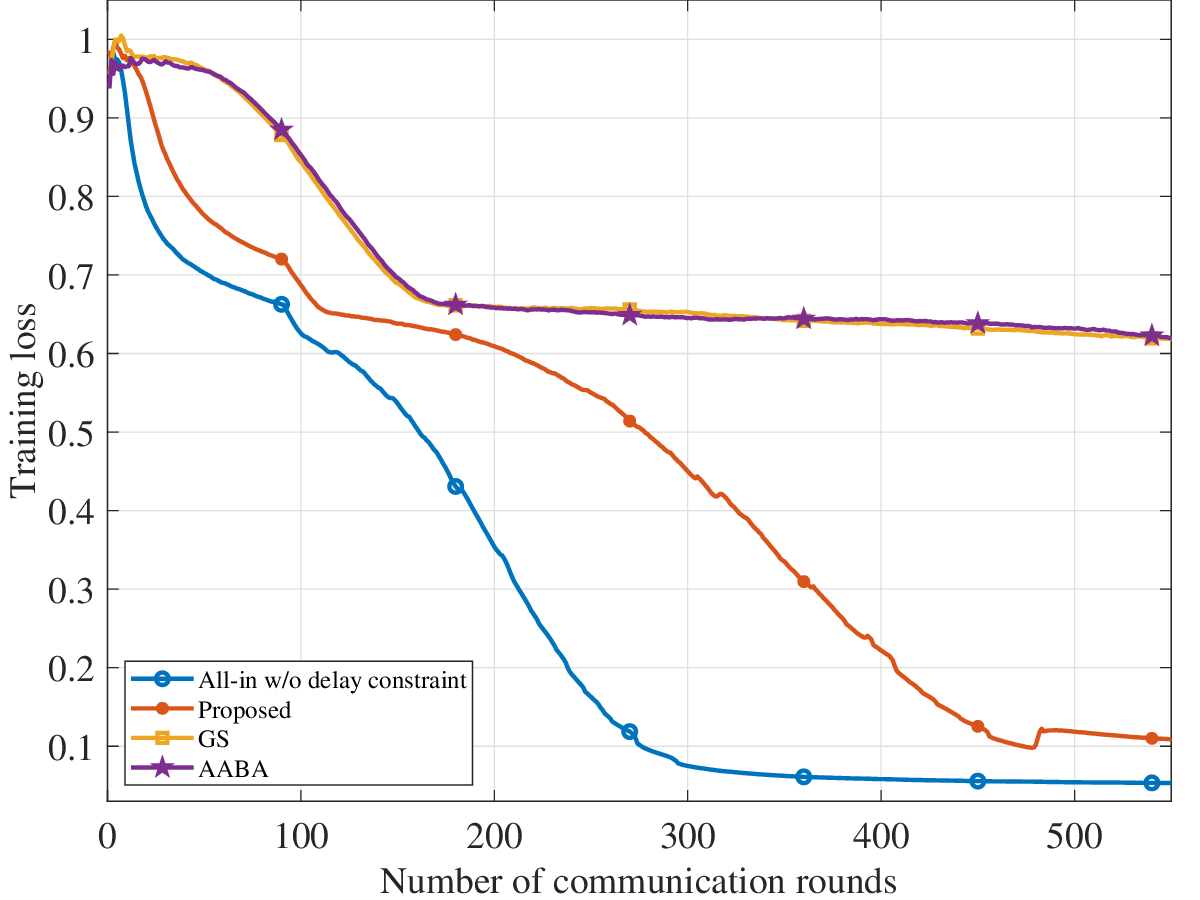}
			\label{fig2_1}
		\end{minipage}
	}
	\subfigure[Training loss (QNLI)]
	{
		\begin{minipage}[t]{0.31\textwidth}
			\centering
			\includegraphics[width=1\linewidth]{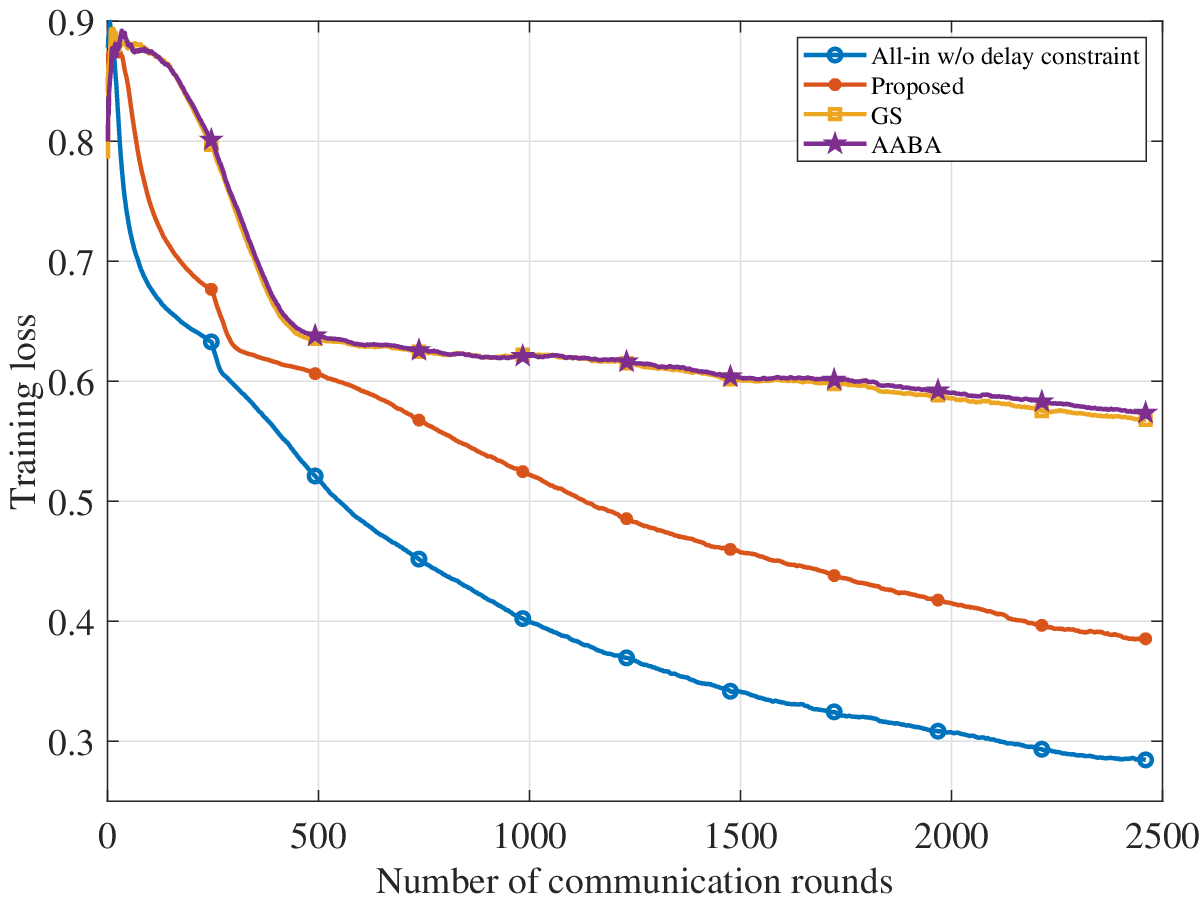}
			\label{fig3_1}
		\end{minipage}
	}
	\\
	\subfigure[Test accuracy (SST-2)]
	{
		\begin{minipage}[t]{0.31\textwidth}
			\centering 
			\includegraphics[width=1\linewidth]{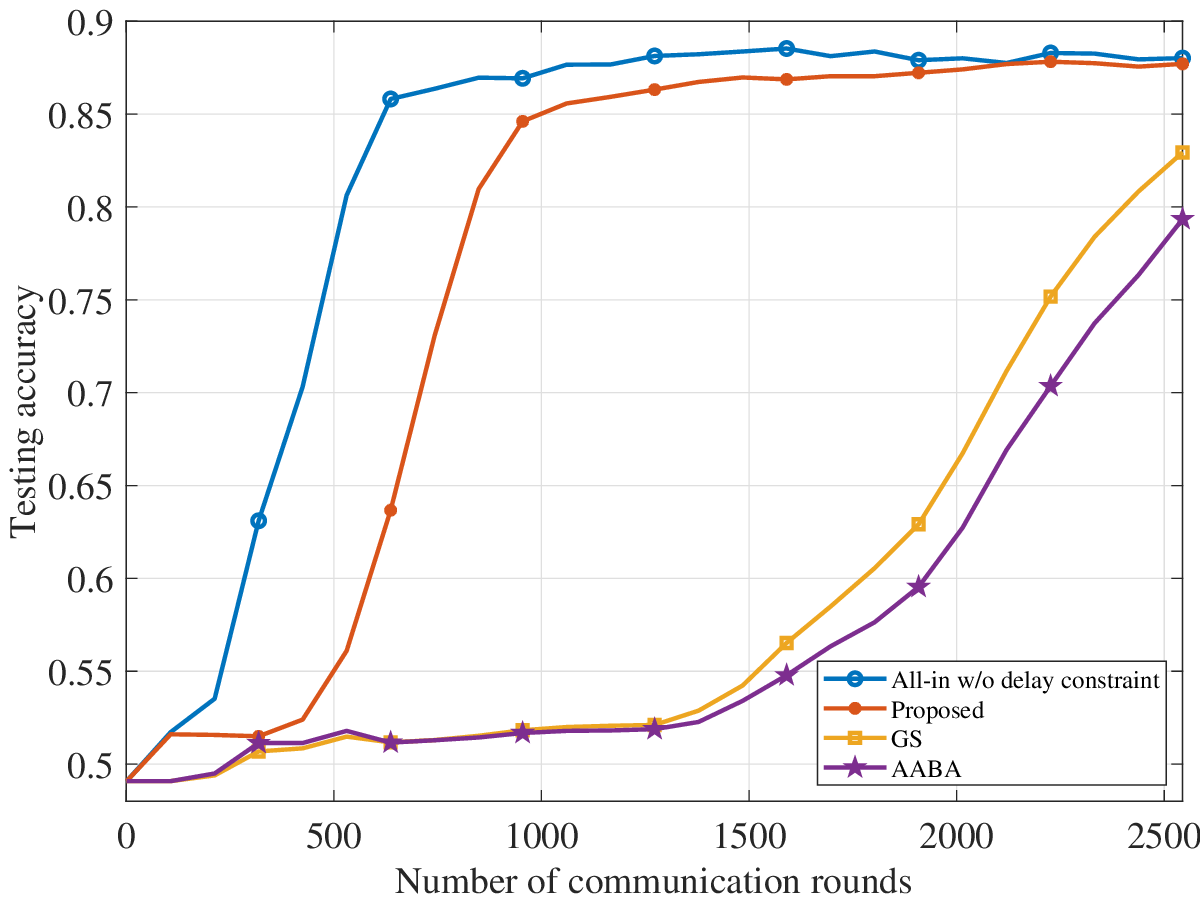}
			\label{fig1_2}
		\end{minipage}
	}
	\subfigure[Test accuracy (MRPC)]
	{
		\begin{minipage}[t]{0.31\textwidth}
			\centering 
			\includegraphics[width=1\linewidth]{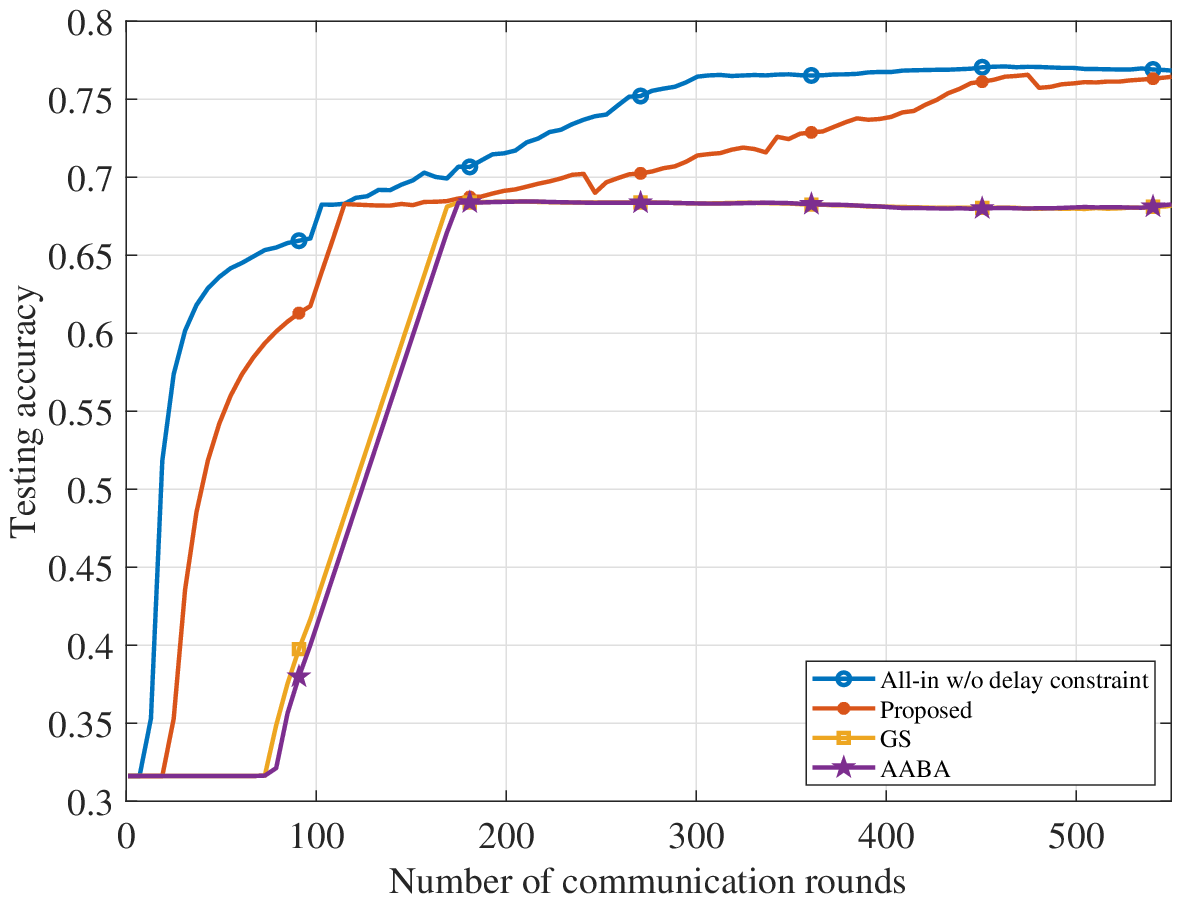}
			\label{fig2_2}
		\end{minipage}
	}
	\subfigure[Test accuracy (QNLI)]
	{
		\begin{minipage}[t]{0.31\textwidth}
			\centering 
			\includegraphics[width=1\linewidth]{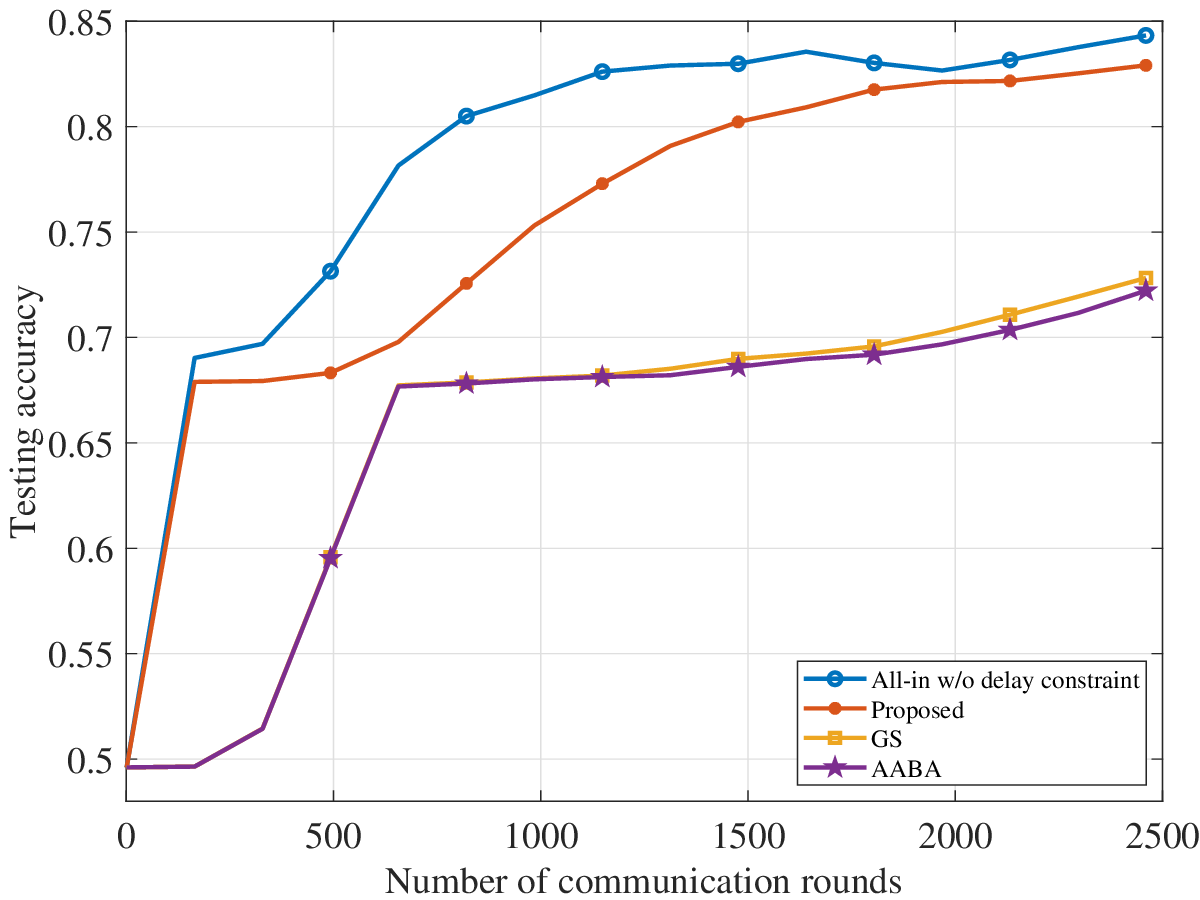}
			\label{fig3_2}
		\end{minipage}
	}
	\caption{Training loss and test accuracy versus the number of communication rounds.}
	\label{fig1}
\end{figure*}

\textbf{FedFT Setting}: We adopt BERT as the pre-trained FM, which has more than $100$ million parameters. We then apply LoRA to the value, key, queue matrices and the dense layers in the pooler of the encoder. Specifically, we set the rank of the low-rank matrices in LoRA as $8$, and the associated weight parameter as $16$. 
As a result, the number of trainable parameters is $454$K and is $0.42\%$ of the total number of parameters. 
We adopt SST-2, MRPC, and QNLI of the GLUE\cite{wang2018glue} benchmark for FedFT. 
Specifically, SST-2 is used for the language sentiment analysis and consists of $67350$ and $1821$ samples in the training set and testing set, respectively.
MRPC is used for semantic texture similarity and consists of $4076$ and $1725$ samples in the training set and testing set, respectively.
QNLI is used for natural language inference and consists of $104743$ and $5463$ samples in the training set and testing set, respectively.
We assign each edge device with $5\%$ data evenly sampled from the global dataset, and deploy $K = 20$ edge devices in total to perform FedFT. Besides, each edge device participates in the training using a mini-batch data with $32$ samples.
We set the learning rate $\eta$ as $1\times 10^{-4}$ and total communication rounds as $5000$.

\textbf{Benchmarks}: we compare the proposed online algorithm with the following benchmarks.
\begin{enumerate}
    \item \textbf{All-in w/o delay constraint}: All edge devices participate in the FedFT without considering the average latency constraint. This represents the best learning performance of the considered FedFT system.
    \item \textbf{Adaptive average bandwidth allocation (AABA)}: Differing from SABA, the total bandwidth is equally allocated to all scheduled edge devices, i.e., $B_k(t) = \frac{B}{|\tilde{\mathcal{K}}(t)|}, \forall k\in\tilde{\mathcal{K}}(t)$. Consequently, the scheduling set $\tilde{\mathcal{K}}(t)$ is determined by
    \begin{equation}
        k = \arg\max_{\in\mathcal{K}}\left\{ \frac{N\mu}{B\log_2(1+\frac{Nk|h_k(t)|^2}{B\sigma^2})}\leq \bar{D}\right\}.
    \end{equation}
    \item \textbf{Greedy strategy (GS)}: By replacing the average constraint with the maximum transmission delay constraint (i.e., $D(t) \leq \bar{D}$), the required bandwidth for each edge device can be obtained by bisection searching the optimal solution for the following equation
    \begin{equation}
        \frac{\mu}{B_k(t)\log_2(1+\frac{|h_k(t)|^2}{B_k(t)\sigma^2})} =\bar{D}.
    \end{equation}
    In particular, through adding the edge device into the scheduling set one by one according to \eqref{reorder_channel_gain}, the terminate condition is met when the residual bandwidth cannot support one more edge device to participate in gradient aggregation.
\end{enumerate}

\subsection{Convergence Performance for Large Language Model}

Fig. \ref{fig1} demonstrates the learning performance comparison between the proposed algorithm and the benchmarks from the perspective of training loss and testing accuracy. 
Beyond the best performance for the \textbf{All-in w/o delay constraint} scheme, the proposed online algorithm outperforms the rest of benchmarks with significantly gap.
It is worth noting that Fig. \ref{fig1} validates the analytical result in Theorem \ref{thm1} as the algorithms that can schedule more devices correspond to better performance.
Specifically, as shown in Fig. \ref{fig1_1}, the training loss corresponding to the proposed algorithm declines rapidly around the $1500$-th communication round and then gradually converges in the rest of communication rounds. Meanwhile, the changes of the training loss corresponding to the \textbf{GS} and \textbf{AABA} is not significant.
This is because, the proposed online algorithm can adaptively schedule more edge devices with exceeding the transmission delay budget under good channel condition and compensate for it in the future communication rounds.
Consequently, compared with the \textbf{GS} and \textbf{AABA}, more edge devices on the average can participate in the gradient aggregation under the scheduling policy of the proposed algorithm, while the \textbf{GS} and \textbf{AABA} are solving problem \textbf{P1} without a long-term view but regarding it as one-shot problems. 
On the other hand, the convergence performance corresponding to \textbf{GS} is better than that of \textbf{AABA}.
This is because, the bandwidth allocation based on the \textbf{GS} algorithm is optimized according to the channel gain, where that of the \textbf{AABA} is averaged among the scheduled edge devices and results in potential radio resource misallocation.

Fig. \ref{fig1_2} presents the testing accuracy comparison versus different communication rounds. 
In particular, the rapid improvement of the test accuracy for \textbf{All-in w/o delay constraint} demonstrates the importance of large number of edge devices involved in FedFT.
The proposed algorithm achieves closed performance to that of the \textbf{All-in w/o delay constraint} scheme under limited radio resource.
As observed, the testing accuracy corresponding to the proposed algorithm increase slowly in the early stage, and then increasing faster in the later rounds, which validates the analysis in Section \ref{SecPro}. 
By effectively allocate bandwidth from a long-term perspective, the proposed online algorithm excels the rest of benchmarks that takes the limited radio resource into account. Besides, the achieved testing accuracy correspond to \textbf{GS} is higher than that of \textbf{AABA} as expected.
\subsection{Convergence Performance for Large Visual Model}
In the following, we deploy the proposed federated LoRA framework over the popular LVM, vision transformer (ViT)\cite{dosovitskiy2020vit}, which consists of $86.5$ million total neural parameters distributed in $12$ stacked transformers. 
Specifically, we implement the proposed framework over the encoder of ViT with setting the rank $r = 8$, where the number of trainable parameters is $663$ thousand accounting for $0.76\% $ of the total parameters.
We evaluate the learning performance on the CIFAR-10 dataset, which is constituted with $60000$ color images in $10$ types.
\begin{figure}[t]
	\centering
	\subfigure[Training loss (CIFAR-10)]
	{
		\begin{minipage}[t]{0.43\textwidth}
			\centering
			\includegraphics[width=0.9\linewidth]{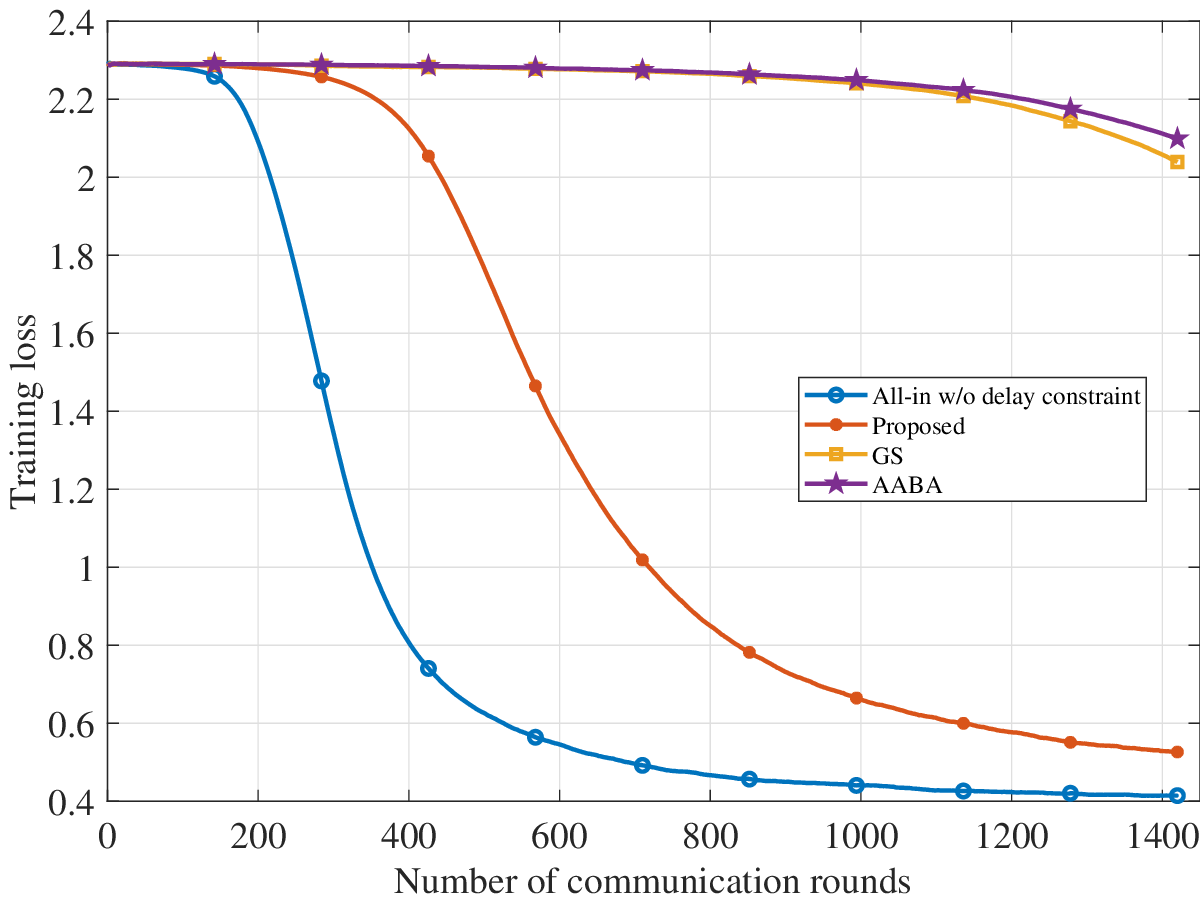}
			\label{fig4_1}
		\end{minipage}
	}
	\subfigure[Test accuracy (CIFAR-10)]
	{
		\begin{minipage}[t]{0.43\textwidth}
			\centering 
			\includegraphics[width=0.9\linewidth]{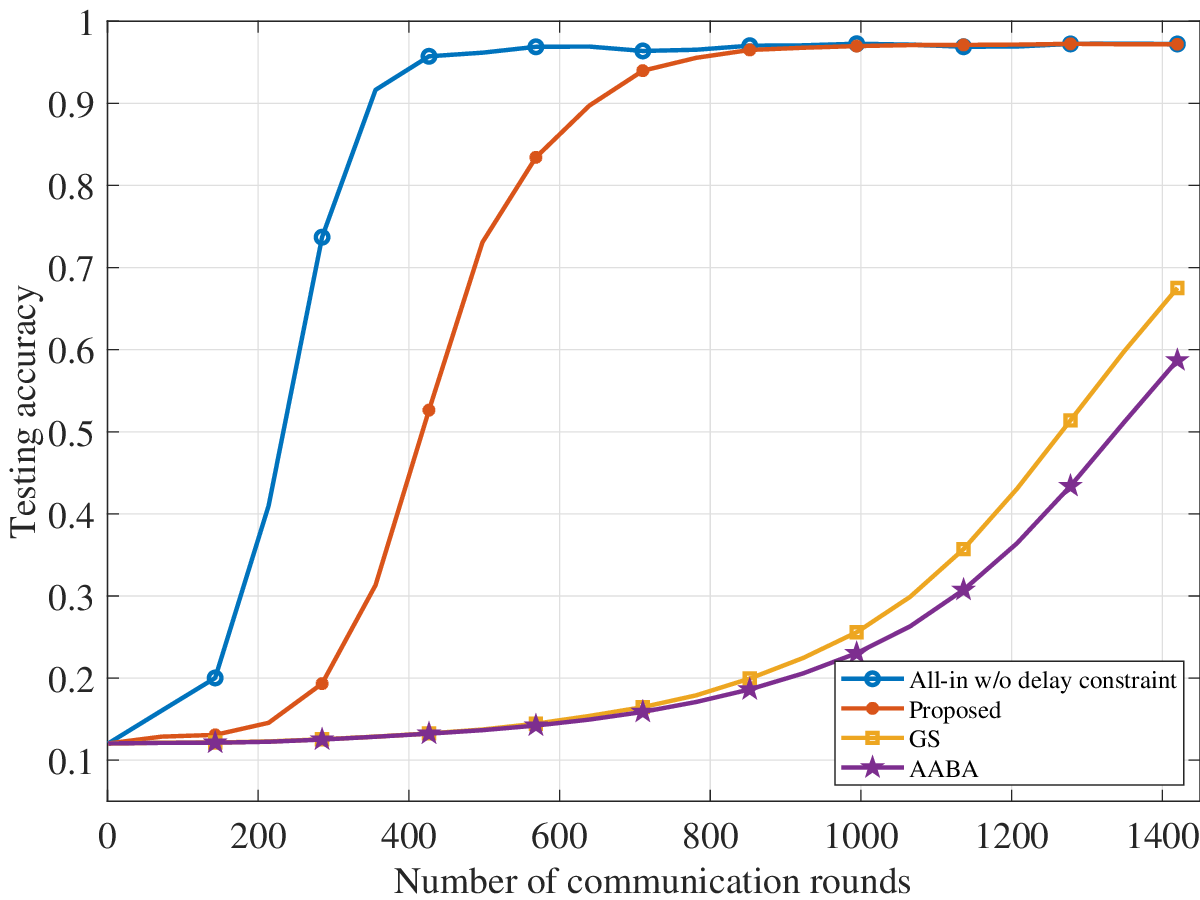}
			\label{fig4_2}
		\end{minipage}
	}
	\caption{Training loss and test accuracy versus number of communication rounds.}
	\label{fig4}
\end{figure}

Fig. \ref{fig4} shows the learning performance comparison of the LLM under different datasets versus number of communication rounds, where the effectiveness of the proposed algorithm is validated via the close performance to that of the All-in scheme in terms of the training loss in Fig. \ref{fig4_1} and test accuracy in Fig. \ref{fig4_2}.
Specifically, the proposed algorithm achieves significantly lower training loss than that of the \textbf{GS} and \textbf{AABA} schemes, demonstrating the superiority of the proposed algorithm aspects of the convergence behavior. 
The test accuracy comparison in Fig. \ref{fig4_2} also validates this observation, where the proposed algorithm achieves higher test accuracy than that of benchmarks considering device scheduling.
Moreover, it can be observed that the curves in Fig. \ref{fig4} is much smoother than that in Fig. \ref{fig1}. This is because, the size of CIFAR-10 dataset is much larger than that of MRPC, SST-2, and QNLI datasets.
Consequently, more samples are assigned to each edge device for local update, which reduces the variance of local gradients and averaged gradient of low-rank matrices.
\begin{figure}[t]
	\centering
	\subfigure[Training loss (CIFAR-10)]
	{
		\begin{minipage}[t]{0.43\textwidth}
			\centering
			\includegraphics[width=0.9\linewidth]{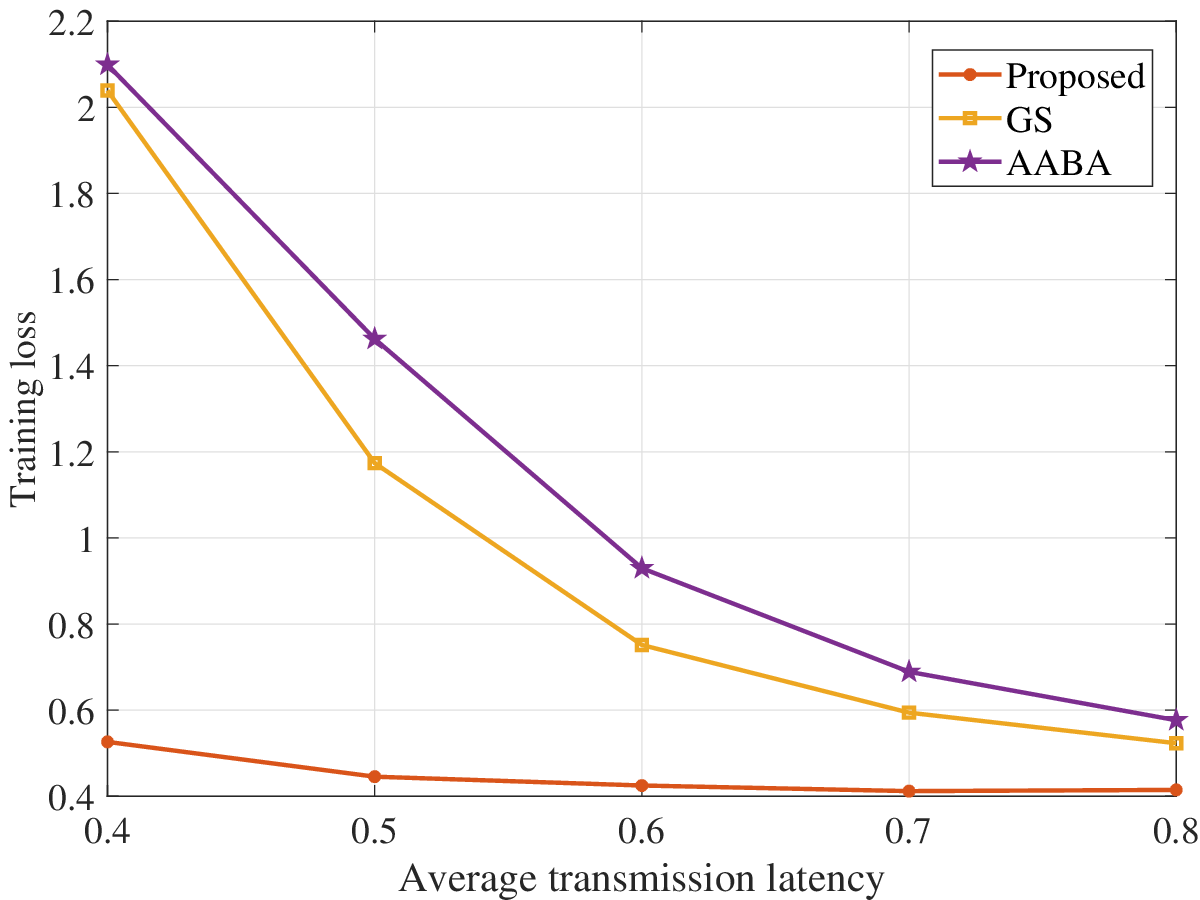}
			\label{fig5_1}
		\end{minipage}
	}
	\subfigure[Test accuracy (CIFAR-10)]
	{
		\begin{minipage}[t]{0.43\textwidth}
			\centering 
			\includegraphics[width=0.9\linewidth]{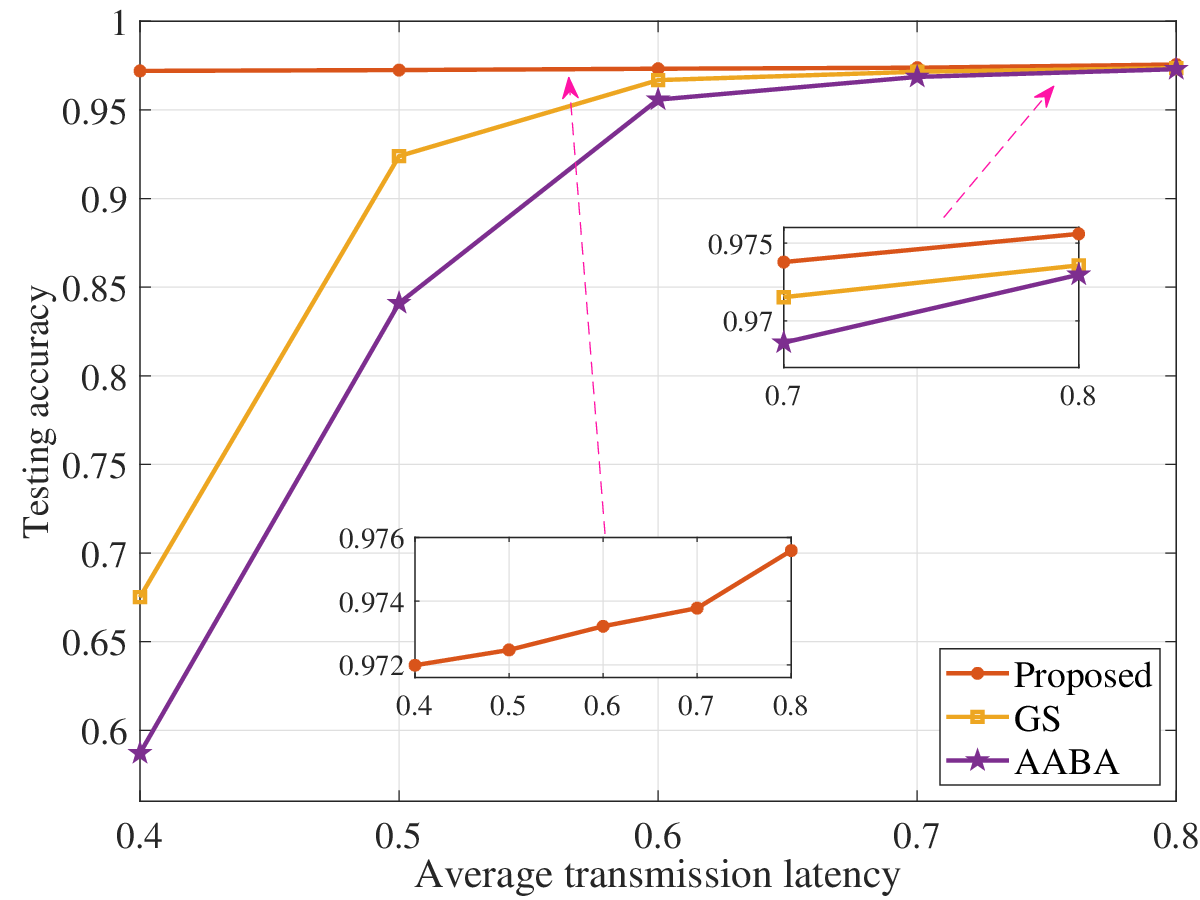}
			\label{fig5_2}
		\end{minipage}
	}
	\caption{Training loss and test accuracy versus the average transmission latency constraint.}
	\label{fig5}
\end{figure}
\subsection{Case Study: Latency Constraint}
With the ViT, we further make the learning performance comparison under different average latency constraint, as shown in Fig. \ref{fig5}.
As observed, the proposed algorithm achieves the best learning performance among the benchmarks considering device scheduling under different latency constraint.
Specifically, as $\bar{D}$ increases, more edge device can be scheduled for gradient aggregation, thereby enhancing the convergence performance and the test accuracy thereafter.
This trend can be observed in Fig \ref{fig5_1} and \ref{fig5_2}, respectively.
In particular, due to the marginal utility of increasing number of edge devices, the reduction of training loss of the proposed algorithm is not as significant as that of the \textbf{GS} and \textbf{AABA} scheme, which also demonstrate that the proposed algorithm can achieve desired convergence behavior under stringent latency constraint.
On the other hand, Fig. \ref{fig5_2} shows the superiority of the proposed algorithm compared with the \textbf{GS} and \textbf{AABA} scheme under low average latency constraint, which shows the high flexibility of the proposed algorithm under ever-changing wireless environment. On the contrary, the compared benchmarks schedule the edge devices without long-term perspectives, thereby performing poor under stringent latency constraint.
\begin{figure}[t]
	\centering
	\includegraphics[width=0.8\linewidth]{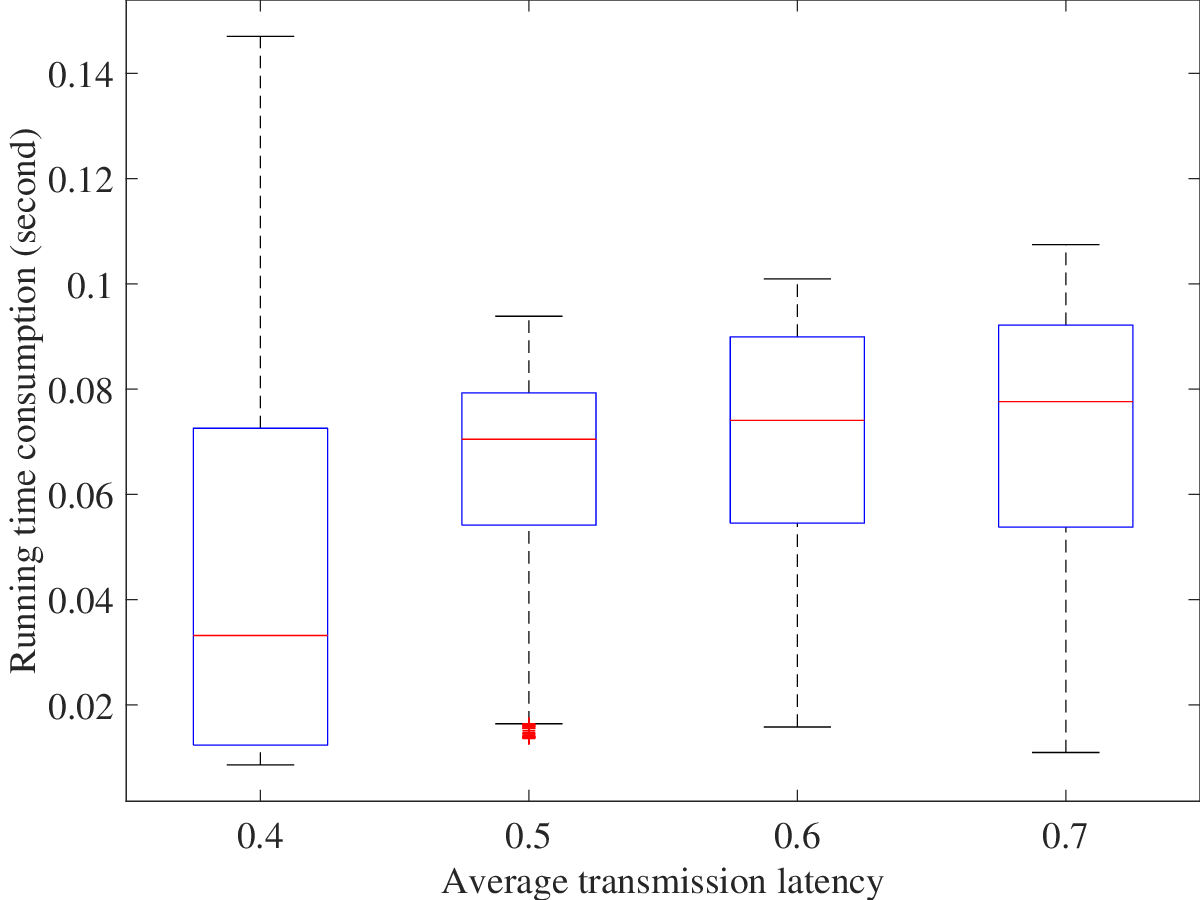}
	\caption{Running time consumption of the proposed algorithm under different average latency constraints.}
	\label{fig6}
\end{figure}
\subsection{Computation Efficiency}
We investigate the calculation efficiency of the proposed algorithm under different transmission latency constraints, as shown in Fig. \ref{fig6}.
As observed, the average running time that is represented as the red lines at the middle of each box is increasing with $\bar{D}$.
This is because, when the average transmission latency constraints looses, more edge devices can be scheduled to participate in the gradient aggregation.
Recall that the proposed algorithm adopts the set expansion strategy, the increasing number of scheduled edge devices in turn increases the average running time consumption.
Besides, when $\bar{D}>0.5$, the maximum time consumption increases along with $\bar{D}$. This is because, scheduling more edge devices requires performing more bisection for optimal $\check{D}(t)$.
It is worth noting that when $\bar{D}=0.4$, the maximum time consumption is greater than that of the other cases. This is because, the stringent latency constraint limits the number of edge devices that participates in the gradient aggregation, which magnifies the impact due to randomness of fading channels and the fluctuation of the virtual queue.
As a result, the proposed algorithm  schedules edge devices with either a small or large number.
This issue is alleviated as $\bar{D}$ increases.
\section{Conclusion}\label{SecCon}
In this paper, we presented a LoRA-based FedFT framework over wireless networks, which reserves the embedding and task modules at the edge devices, while deploying the encoder at the edge server, interconnected via the wireless network. 
We rigorously analyzed the upper bound of the convergence gap for the proposed FedFT and formulated a long-term optimization problem with respect to device scheduling and bandwidth allocation. 
Using Lyapunov analysis, we decomposed the long-term optimization problem into a series of sub-problems and developed an online algorithm employing a set expansion strategy to solve each sub-problem, where the $\Delta$-optimality of the proposed algorithm was proved. 
Simulation results demonstrated the superior learning performance of our proposed algorithm compared to benchmark approaches in the considered FedFT framework.
\section*{Appendix}\label{SecApp}
\subsection{Proof of Lemma \ref{lmm1}}
According to \eqref{EFT_obj}, we can expand the formula as
\begin{align*}
	\Vert\nabla F(\bm{W})\Vert^2=& \left\Vert\frac{1}{K}\sum_{k = 1}^{K} \nabla f_{k}(\tilde{\bm{w}}_{k})\right\Vert^2\\
	=&\left\Vert\frac{1}{K}\sum_{k=1}^{K}\left(\frac{\partial f_k(\bm{w}^{\mathrm{a}})}{\partial \bm{w}^{\mathrm{t}}_{k}}, \frac{\partial f_k(\bm{w}^{\mathrm{t}}_{k})}{\partial \bm{w}^{\mathrm{a}}}\right)\right\Vert^2.
\end{align*}
Then we derive the derivatives for $\bm{w}^{\mathrm{a}}$ and $\{\bm{w}^{\mathrm{t}}_{k}\}$ separately to get
\begin{align*}
	&\Vert\nabla F(\bm{W})\Vert^2\\
	=& \frac{1}{K^2}\left\Vert\sum_{k=1}^{K}\frac{\partial f_{k} (\bm{w}^{\mathrm{t}}_{k})}{\bm{w}^{\mathrm{a}}}\right\Vert^2 + \frac{1}{K^2}\sum_{k=1}^{K}\left\Vert\frac{\partial f_{k} (\bm{w}^{\mathrm{a}})}{\bm{w}^{\mathrm{t}}_{k}}\right\Vert^2\\
	=& \frac{1}{K^2}\left(\left\Vert\sum_{k\in\mathcal{K}}\nabla_{\bm{w}^{\mathrm{a}}} f_{k}(\bm{w}^{\mathrm{a}}; \bm{w}^{\mathrm{t}}_{k})\right\Vert^2\!\!\!\!+\!\!\!\sum_{k\in\mathcal{K}}\left\Vert\nabla_{\bm{w}^{\mathrm{t}}_{k}} f_{k}(\bm{w}^{\mathrm{t}}_{k};\bm{w}^{\mathrm{a}} )\right\Vert^2\right)\!\!\\
    &\leq\frac{2}{K^2}\left(\sum_{k\in\mathcal{K}}\!\left\Vert\nabla_{\bm{w}^{\rm t}}f_k(\bm{w}^{\rm t};\bm{w}_k^{\rm a})\right\Vert^2\!\!+\!\!\sum_{k\in\mathcal{K}}\!\left\Vert\nabla_{\bm{w}^{\rm a}_{k}}f_k(\bm{w}^{\rm a}_k; \bm{w}^{\rm t})\right\Vert^2\!\right).
\end{align*}
\subsection{Proof of Theorem \ref{thm1}}
According to Lemma \ref{lmm1}, we first fix the parameters of task model and then derive the upper bound of convergence with respect to the global adapter $\bm{w}^{\rm a}$.
Recall that in each communication round, $N$ edge devices participate in the gradient aggregation, according to the global model update rule, we have:
\begin{align}
    \begin{aligned}
        \bm{w}^{\rm a}(t+1) &= \bm{w}^{\rm a}(t) - \frac{\eta}{{N(t)}} \sum_{k\in\tilde{\mathcal{K}}(t)}\bm{g}^{\rm a}_k(t).
    \end{aligned}
\end{align}
To assist the derivation, we define a virtual model $\bm{v}^{\rm a}(t+1)$ that assumes all devices participated in each training round as
\begin{align}
    \begin{aligned}
        \bm{v}^{\rm a}(t+1)=\bm{w}^{\rm a}(t)-\frac{\eta}{K} \sum_{k=1}^{K}\bm{g}^{\rm a}_k(t).
    \end{aligned}
\end{align}
For the simplicity of the following proof, we abbreviate $\nabla_{\bm{w}^{\rm t}} f_k(\bm{w}^{\rm a}_{k}(t);\bm{w}^{\rm t}(t+1) )$ and ${\nabla} \bm{w}^{\rm{a}}_{k} \frac{{\rm{d}}\bm{z}_k(t)}{{\rm{d}}\bm{w}^{\rm{a}}_k(t)}$ as $\nabla_{\bm{w}^{\rm a}_k} f_k(\bm{w}^{\rm a}(t))$ and $\bm{g}^{\rm a}_k(t)$, respectively.

According to Assumption \ref{ass-2}, by defining $\bm{\Upsilon}^{\rm a}(t+1)=\bm{w}^{\rm a}(t+1)-\bm{v}^{\rm a}(t+1)$, we have:
    \begin{align*}\label{proof-main}
        &f_k(\bm{w}^{\rm a}(t+1); \bm{w}^{\rm t}_{k}(t))-f_k(\bm{w}^{\rm a}(t); \bm{w}^{\rm t}_{k}(t))\\
        \leq&\nabla_{\bm{w}^{\rm a}} f_k(\bm{w}^{\rm a}(t))^T\left(\bm{w}^{\rm a}(t+1)-\bm{v}^{\rm a}(t+1) \right)\\
        &+\nabla_{\bm{w}^{\rm a}} f_k(\bm{w}^{\rm a}(t))^T\left(\bm{v}^{\rm a}(t+1)-\bm{w}^{\rm a}(t)\right)\\
        &+\frac{L}{2}\left\Vert\bm{w}^{\rm a}(t+1)-\bm{v}^{\rm a}(t+1)+\bm{v}^{\rm a}(t+1)-\bm{w}^{\rm a}(t)\right\Vert^2\\
        =&\nabla_{\bm{w}^{\rm a}} f_k(\bm{w}^{\rm a}(t))^T\bm{\Upsilon}^{\rm a}(t+1)+\frac{L\eta^2}{2}\left\Vert\frac{1}{K}\sum_{k=1}^K\bm{g}^{\rm a}_k(t)\right\Vert^2\\
        &+\frac{L}{2}\left\Vert\bm{\Upsilon}^{\rm a}(t+1)\right\Vert^2-\eta L\left<\bm{\Upsilon}^{\rm a}(t+1),\frac{1}{K}\sum_{k=1}^K\bm{g}^{\rm a}_k(t)\right>\\
        &-\eta\nabla_{\bm{w}^{\rm a}} f_k(\bm{w}^{\rm a}(t))^T\left(\frac{1}{K}\sum_{k=1}^K\bm{g}^{\rm a}_k(t)\right).
    \end{align*}
By taking the expectation over the random ${N(t)}$ of $K$ edge devices at the $t$-th communication round, we have:
\begin{align}
    \mathbb{E}[\bm{w}^{\rm a}(t+1)]=\bm{v}^{\rm a}(t+1).
\end{align}
Therefore, \eqref{proof-main} can be upper bounded by 
\begin{align*}
    \begin{aligned}
        &\mathbb{E}\left[f_k(\bm{w}^{\rm a}(t+1); \bm{w}^{\rm t}_{k}(t))-f_k(\bm{w}^{\rm a}(t); \bm{w}^{\rm t}_{k}(t))\right]\\
        \leq& -\eta\mathbb{E}\left<\nabla_{\bm{w}^{\rm a}} f_k(\bm{w}^{\rm a}(t)),\frac{1}{K}\sum_{k=1}^K\bm{g}^{\rm a}_k(t)\right>\\
        &+\frac{L}{2}\mathbb{E}\Vert\bm{\Upsilon}^{\rm a}(t+1)\Vert^2 +\frac{L\eta^2}{2K}\sum_{k=1}^K\mathbb{E}\left\Vert\bm{g}^{\rm a}_k(t)\right\Vert^2\\
        =&\eta\mathbb{E}\left<\nabla_{\bm{w}^{\rm a}} f_k(\bm{w}^{\rm a}(t)), \nabla_{\bm{w}^{\rm a}} f_k(\bm{w}^{\rm a}(t))-\frac{1}{K}\sum_{k=1}^K\bm{g}_k^{\rm a}(t)\right>\\
        &+\frac{L\eta^2}{2K}\sum_{k=1}^K\mathbb{E}\Vert\bm{g}_k^{\rm a}(t)\Vert^2+\frac{L}{2}\mathbb{E}\Vert\bm{\Upsilon}^{\rm a}(t+1)\Vert^2 \\
        &-\eta\mathbb{E}\Vert\nabla_{\bm{w}^{\rm a}} f_k(\bm{w}^{\rm a}(t))\Vert^2\\
        =&\frac{\eta}{K}\sum_{k=1}^K\mathbb{E}\left<\nabla_{\bm{w}^{\rm a}} f_k(\bm{w}^{\rm a}(t)),\nabla_{\bm{w}^{\rm a}} f_k(\bm{w}^{\rm a}(t))-\bm{g}_k^{\rm a}(t)\right>\\
        &+\frac{L\eta^2}{2K}\sum_{k=1}^K\mathbb{E}\Vert\bm{g}_k^{\rm a}(t)\Vert^2+\frac{L}{2}\mathbb{E}\Vert\bm{\Upsilon}^{\rm a}(t+1)\Vert^2 \\
        &-\eta\mathbb{E}\Vert\nabla_{\bm{w}^{\rm a}} f_k(\bm{w}^{\rm a}(t))\Vert^2.\\
    \end{aligned}
\end{align*}
By further applying the Cauchy-Schwarz inequality to above inequality, we have 
\begin{equation}\label{upp_bound_exp_wt}
    \begin{split} 
    &\mathbb{E}\left[f_k(\bm{w}^{\rm a}(t+1); \bm{w}^{\rm t}_{k}(t))-f_k(\bm{w}^{\rm a}(t); \bm{w}^{\rm t}_{k}(t))\right]\\
        \leq& \frac{\eta}{2K}\sum_{k=1}^K\mathbb{E}\Vert\nabla_{\bm{w}^{\rm a}} f_k(\bm{w}^{\rm a}(t))-\bm{g}_k^{\rm a}(t)\Vert^2\\
        &+\frac{L\eta^2}{2K}\sum_{k=1}^K\mathbb{E}\Vert\bm{g}_k^{\rm a}(t)\Vert^2+\frac{L}{2}\mathbb{E}\Vert\bm{\Upsilon}^{\rm a}(t+1)\Vert^2 \\
        &+\frac{\eta}{2K}\sum_{k=1}^K\mathbb{E}\Vert\nabla_{\bm{w}^{\rm a}} f_k(\bm{w}^{\rm a}(t))\Vert^2\\
        &-\eta\mathbb{E}\Vert\nabla_{\bm{w}^{\rm a}} f_k(\bm{w}^{\rm a}(t))\Vert^2.\\
    \end{split}
\end{equation}
Recall that $\bm{g}_k(t)$ is an unbiased estimation of $\nabla_{\bm{w}^{\rm t}} f_k(\bm{w}^{\rm t}(t))$, we can further reorganize \eqref{upp_bound_exp_wt} and have 
\begin{equation}\label{upp_bound_exp_wt_2}
    \begin{split} 
    &\mathbb{E}\left[f_k(\bm{w}^{\rm a}(t+1); \bm{w}^{\rm t}_{k}(t))-f_k(\bm{w}^{\rm a}(t); \bm{w}^{\rm t}_{k}(t))\right]\\
        \leq& \frac{L\eta^2}{2K}\sum_{k=1}^K\mathbb{E}\Vert\bm{g}^{\rm a}_k(t)\Vert^2+\frac{L}{2}\mathbb{E}\Vert\bm{\Upsilon}^{\rm a}(t+1)\Vert^2 \\
        &+\frac{\eta}{2K}\sum_{k=1}^K\mathbb{E}\Vert\nabla_{\bm{w}^{\rm a}} f_k(\bm{w}^{\rm a}(t))\Vert^2\\
        &-\eta\mathbb{E}\Vert\nabla_{\bm{w}^{\rm a}} f_k(\bm{w}^{\rm a}(t))\Vert^2.\\
    \end{split}
\end{equation}
To this end, we unfold $\mathbb{E}\Vert\bm{\Upsilon}^{\rm t}(t+1)\Vert^2$ as below
\begin{align}\label{upsilon}
    \begin{aligned}
        &\mathbb{E}\left\Vert\bm{\Upsilon}^{\rm a}(t+1)\right\Vert^2
        =\mathbb{E}\left\Vert\frac{1}{{N(t)}}\sum_{k\in\tilde{\mathcal{K}}(t)}\bm{g}_k^{\rm a}(t)-\bm{g}(t)\right\Vert^2\\
        =&\mathbb{E}\left\Vert\frac{1}{{N(t)}}\sum_{k=1}^K\mathbb{I}(k\in\tilde{\mathcal{K}}(t))(\bm{g}_k^{\rm a}(t)-\bm{g}(t))\right\Vert^2\\
        =&\frac{1}{{N(t)}^2}\mathbb{E}\left[\sum_{k=1}^K\mathbb{I}(k\in\tilde{\mathcal{K}}(t))\Vert\bm{g}_k^{\rm a}(t)-\bm{g}(t)\Vert^2\right.\\
        &\left.+\sum_{k=1}^K\sum_{\substack{j=1\\j\neq k}}^K\mathbb{I}(k\in\tilde{\mathcal{K}}(t))\mathbb{I}(j\in\tilde{\mathcal{K}}(t))\right.\\
        &\left.\times \left<\bm{g}_k^{\rm a}(t)-\bm{g}(t),\bm{g}^{\rm a}_{j}(t)-\bm{g}(t)
        \right>\right],
    \end{aligned}
\end{align}
where $\bm{g}^{\rm a}(t) =\frac{1}{K} \sum_{k=1}^{K}\bm{g}^{\rm a}_k(t)$, $\mathbb{I}(k\in\tilde{\mathcal{K}}(t))$ denotes an indicator function that equals to 1 when $k\in\tilde{\mathcal{K}}(t) $.
Since the expectation is performed over $\tilde{\mathcal{K}}(t)$, the first part of \eqref{upsilon} can be represented as
\begin{align}\label{upsilon-1}
    \begin{aligned}
        &\frac{1}{{N(t)}^2}\frac{\mathcal{C}_{K-1}^{{N(t)}-1}}{\mathcal{C}_{K}^{N(t)}}\sum_{k=1}^{N(t)}\Vert\bm{g}^{\rm a}_k(t)-\bm{g}(t)\Vert^2\\
        =&\frac{1}{{N(t)}K}\sum_{k=1}^{N(t)}\Vert\bm{g}^{\rm a}_k(t)-\bm{g}(t)\Vert^2,
    \end{aligned}
\end{align}
and we perform similar derivations as in appendix B in \cite{shi2020joint} and appendix C in \cite{amiri2021convergence}, the second part can be rewritten as
\begin{align}\label{upsilon-2}
    \begin{aligned}
        &\frac{1}{{N(t)}^2}\mathbb{E}_{\tilde{\mathcal{K}}(t)}\sum_{k=1}^K\sum_{j=1,j\neq k}^K\mathbb{I}(k\in\tilde{\mathcal{K}}(t))\mathbb{I}(j\in\tilde{\mathcal{K}}(t))\\
        &\times \left<\bm{g}^{\rm a}_k(t)-\bm{g}(t),\bm{g}^{\rm a}_{j}(t)-\bm{g}(t)\right>\\
        &=\frac{\mathcal{C}_{K-2}^{{N(t)}-2}}{\mathcal{C}_{K}^{N(t)}}\sum_{k=1}^K\sum_{j=1,j\neq k}^K\left<\bm{g}^{\rm a}_k(t)-\bm{g}(t),\bm{g}^{\rm a}_{j}(t)-\bm{g}(t)\right>.
    \end{aligned}
\end{align}
With (\ref{upsilon-1}) and (\ref{upsilon-2}), we can rewrite (\ref{upsilon}) as
\begin{align*}
    \begin{aligned}
    \mathbb{E}\left\Vert\bm{\Upsilon}^{\rm t}(t+1)\right\Vert^2
    &=\frac{K-{N(t)}}{K{N(t)}(K-1)}\sum_{k=1}^K\mathbb{E}\Vert\bm{g}^{\rm a}_k(t)-\bm{g}(t)\Vert^2\\
    &\leq\frac{K-{N(t)}}{K{N(t)}(K-1)}\sum_{k=1}^K\mathbb{E}\Vert\bm{g}^{\rm a}_k(t)\Vert^2,
    \end{aligned}
\end{align*}
where the inequality is due to the fact that $\sum_{i=1}^{N(t)}\Vert\bm{a}_i-\Bar{\bm{a}}\Vert^2\leq\sum_{i=1}^{N(t)}\Vert\bm{a}_i\Vert^2$.
Hence, we can reorganize \eqref{upp_bound_exp_wt} and have
\begin{equation}
    \begin{split}
        &\mathbb{E}\left[f_k(\bm{w}^{\rm a}(t+1); \bm{w}^{\rm t}_{k}(t))-f_k(\bm{w}^{\rm a}(t); \bm{w}^{\rm t}_{k}(t))\right]\\
        \leq&\frac{L\eta^2}{2K}\sum_{k=1}^K\mathbb{E}\Vert\bm{g}^{\rm a}_k(t)\Vert^2+\frac{(K-{N(t)})L}{2K{N(t)}(K-1)}\sum_{k=1}^K\mathbb{E}\Vert\bm{g}^{\rm a}_k(t)\Vert^2\\
        &+\frac{\eta}{2K}\sum_{k=1}^K\mathbb{E}\Vert\nabla_{\bm{w}^{\rm a}} f_k(\bm{w}^{\rm a}(t))\Vert^2\\
        &-\eta\mathbb{E}\Vert\nabla_{\bm{w}^{\rm a}} f_k(\bm{w}^{\rm a}(t))\Vert^2.\\
    \end{split}
\end{equation}
Recall that the update of task model at each edge device is conducted without averaging operation, we can derive the corresponding upper bound with respect to the task model as below
\begin{equation}
    \begin{split}
        &\mathbb{E}\left[f_k(\bm{w}^{\rm t}(t+1); \bm{w}^{\rm a}_{k}(t))-f_k(\bm{w}^{\rm t}(t); \bm{w}^{\rm a}_{k}(t))\right]\\
        \leq&\frac{L\eta^2}{2K}\sum_{k=1}^K\mathbb{E}(\Vert\bm{g}^{\rm t}_k(t)\Vert^2)-\eta\mathbb{E}\Vert\nabla_{\bm{w}^{\rm t}} f_k(\bm{w}^{\rm t}(t))\Vert^2\\
        &+\frac{\eta}{2K}\sum_{k=1}^K\mathbb{E}\Vert\nabla_{\bm{w}^{\rm t}} f_k(\bm{w}^{\rm t}(t))\Vert^2.\\
    \end{split}
\end{equation}
To this end, we can combine these two separate bounds as following
\begin{equation*}
    \begin{split}
        &F(\bm{W}(t+1))- F(\bm{W}(t))\\
        \leq&\sum_{k=1}^{K}\left[f_k(\bm{w}^{\rm a}(t+1); \bm{w}^{\rm t}_{k}(t))-f_k(\bm{w}^{\rm a}(t); \bm{w}^{\rm t}_{k}(t))\right.\\
        &\left.+ f_k(\bm{w}^{\rm a}(t+1); \bm{w}^{\rm t}_{k}(t+1))-f_k(\bm{w}^{\rm a}(t+1); \bm{w}^{\rm t}_{k}(t))  \right]\\
        \leq&\sum_{k=1}^{K}\left[\frac{L\eta^2}{2K}\sum_{k=1}^K\mathbb{E}\left(\Vert\bm{g}_k^{\rm t}(t)\Vert^2 + \Vert\bm{g}^{\rm a}_k(t)\Vert^2\right)\right.\\
        &+\frac{(K-{N(t)})L}{2K{N(t)}(K-1)}\sum_{k=1}^K\mathbb{E}(\Vert\bm{g}^{\rm a}_k(t)\Vert^2)\\
        &+\frac{\eta}{2K}\sum_{k=1}^{K}\left(\mathbb{E}(\Vert\nabla_{\bm{w}^{\rm t}} f_k(\bm{w}^{\rm t}(t))\Vert^2) + \mathbb{E}(\Vert\nabla_{\bm{w}^{\rm a}} f_k(\bm{w}^{\rm a}(t))\Vert^2)\right)\\
        &\left.-\eta\left(\mathbb{E}(\Vert\nabla_{\bm{w}^{\rm t}} f_k(\bm{w}^{\rm t}(t))\Vert^2) +\mathbb{E}(\Vert\nabla_{\bm{w}^{\rm a}} f_k(\bm{w}^{\rm a}(t))\Vert^2)\right)\right]\\
            \end{split}
    \end{equation*}
        \begin{equation*}
        	\begin{split}
        \mathop{\leq}^{(a)}&\sum_{k=1}^{K}\left[\frac{L\eta^2}{2}\left(\Vert\nabla_{\bm{w}^{\rm t}} f_k(\bm{w}^{\rm t}(t))\Vert^2 + \Vert\nabla_{\bm{w}^{\rm a}} f_k(\bm{w}^{\rm a}(t))\Vert^2\right)\right.\\
        &+\frac{(K-{N(t)})L}{2{N(t)}(K-1)}(\Vert\nabla_{\bm{w}^{\rm a}} f_k(\bm{w}^{\rm a}(t))\Vert^2)+\frac{L\eta^2}{2}\phi^2(\Omega^{\rm a}+\Omega^{\rm t})\\
        &+\frac{\eta}{2}\left(\Vert\nabla_{\bm{w}^{\rm t}} f_k(\bm{w}^{\rm t}(t))\Vert^2 + \Vert\nabla_{\bm{w}^{\rm a}} f_k(\bm{w}^{\rm a}(t))\Vert^2\right)\\
        & -\eta \left( \Vert\nabla_{\bm{w}^{\rm t}} f_k(\bm{w}^{\rm t}(t))\Vert^2 + \Vert\nabla_{\bm{w}^{\rm a}} f_k(\bm{w}^{\rm a}(t))\Vert^2\right) \\
&\left.-\eta\phi^2(\Omega^{\rm a}+\Omega^{\rm t}) + \frac{(K-{N(t)})L}{2{N(t)}(K-1)}\phi^2\Omega^{\rm a}\right]\\    =&\sum_{k=1}^{K}\left[\left(\frac{L\eta^2-\eta}{2}+\frac{(K-{N(t)})L}{2{N(t)}(K-1)}\right)\Vert\nabla_{\bm{w}^{\rm a}} f_k(\bm{w}^{\rm a}(t))\Vert^2\right.\\
        &+\frac{L\eta^2-\eta}{2}\Vert\nabla_{\bm{w}^{\rm t}} f_k(\bm{w}^{\rm t}(t))\Vert^2+\phi^2\Omega^{\rm t}\left(L\eta^2 
        -\eta\right) \\
        &\left.+\phi^2\Omega^{\rm a}\left(L\eta^2 
        -\eta+\frac{(K-{N(t)})L}{2{N(t)}(K-1)}\right)\right],
    \end{split}
\end{equation*}
where $(a)$ is because $\mathbb{E}(\Vert \bm{x}\Vert^2) = \sum\text{var}(\bm{x}_i) + \Vert\mathbb{E}(\bm{x})\Vert^2$, and Assumption \ref{ass-4}, $\Omega^{\rm a}$ and $\Omega^{\rm t}$ denote the number of elements in $\Vert\nabla_{\bm{w}^{\rm a}} f_k(\bm{w}^{\rm a}(t))\Vert^2 $ and $ \Vert\nabla_{\bm{w}^{\rm t}} f_k(\bm{w}^{\rm t}(t))\Vert^2$, respectively.

Since we have $\frac{(K-{N(t)})L}{2{N(t)}(K-1)}\geq 0$, we can define 
\begin{align*}
    &\varsigma(t)=-\frac{L\eta^2-\eta}{K^2}-\frac{(K-{N(t)})L}{2{N(t)}(K-1)K^2},\\
    &\alpha(t)=\phi^2K^2\varsigma(t)\Omega^{\rm a}, \quad\beta = \Omega^{\rm t}\left(L\eta^2 
    -\eta\right)\phi^2.
\end{align*}
To this end, when $L<\frac{\eta {N(t)}(K-1)}{(K-1){N(t)}\eta^2+K-{N(t)}}$, the following inequality can be derived
\begin{equation*}
    \begin{split}
        &F(\bm{W}(T+1))- F(\bm{W}(T))\\
        \leq& -\frac{K^2 \varsigma(T)}{2}\sum_{k=1}^{K}\left(\Vert\nabla_{\bm{w}^{\rm t}} f_k(\bm{w}^{\rm t}(T))\Vert^2\!+\!\Vert\nabla_{\bm{w}^{\rm a}} f_k(\bm{w}^{\rm a}(T))\Vert^2\right)\\
        &-\alpha(T)+\beta\\
       \mathop{\leq}^{(b)} &-\varsigma(T)\left\Vert\nabla F(\bm{W}(T))\right\Vert^2 - \alpha(T)+\beta\\
        \mathop{\leq}^{(c)}& -2\tau \varsigma(T)(F(\bm{W}(T)) - F(\bm{W}^\star)) - \alpha(T)+\beta,
    \end{split}
\end{equation*}
where $(b)$ is due to Lemma \ref{lmm1} and $(c)$ is due to the Assumption \ref{ass-5}.
Thus, we can have 
\begin{equation*}
    \begin{split}
        &F(\bm{W}(T+1))- F(\bm{W}^{\star})\\
        \leq&\left(1-2\tau \varsigma(T)\right)(F(\bm{W}(T)) - F(\bm{W}^\star))+\beta-\alpha(T)\\
        =&(1-2\tau \varsigma(T))\left[F(\bm{W}(T)) - F(\bm{W}(T-1)) \right.\\
        &\left.+ F(\bm{W}(T-1)) -   F(\bm{W}^\star)\right]+\beta-\alpha(T)\\
        \leq&\prod_{i=0}^{T}(1-2\tau \varsigma(i))(F(\bm{W}(0)) -F(\bm{W}^\star))+\beta-\alpha(T) \\
        &+  \sum_{i=1}^{T}\left(\prod_{j=0}^{i-1}\left(1-2\tau \varsigma(T-j)\right) \right)(\beta-\alpha(T-i)).
    \end{split}
\end{equation*}
\subsection{Proof of Lemma \ref{lmm2}}
By substituting \eqref{lya_fun} into \eqref{lya_drift}, we have 
\begin{align*}
	\Delta V\left(\hat{D}(t)\right)&= \mathbb{E}\left[V\left(\hat{D}(t+1)\right) - V\left(\hat{D}(t)\right)\;|\; \hat{D}(t)\right]\\
	&=\frac{1}{2}\mathbb{E}\left[\hat{D}(t)^2(t+1) - \hat{D}(t)^2_{k}(t)\right].
\end{align*}
Then we expand quadratic terms and combine like terms, the above equation can be rewritten as
\begin{align*}
	\Delta V\left(\hat{D}(t)\right)
	=&\frac{1}{2}\left[\left(\hat{D}(t)+D(t)-\bar{D}(t) \right)^2 - \hat{D}(t) \right]\\
	=& \hat{D}(t)\left(D(t)-\bar{D}(t)\right)-D(t)\bar{D}(t)\\
    \mathop{\leq}^{(d)}&\hat{D}(t)\left(D(t)-\bar{D}(t)\right),
\end{align*}
where $(d)$ is due to the term $D(t)\bar{D}(t)$ is non-negative.
\subsection{Proof of Lemma \ref{lmm3}}
By denoting $\frac{1}{B_k(t)}$ by $y_k(t)$, we can rewrite \eqref{Pro_4} as 
\begin{equation}\label{trans_equation}
    2^{(\mu/\check{D}(t))y_k} = \frac{|h_k(t)|^2}{\sigma^2}y_k(t)+1,
\end{equation}
which is a transcendental equation due to the exponential term.
By multiplying $-\nu(t) = \frac{\mu|h_k(t)|^2}{\check{D}(t)\sigma^2}$ and dividing $2^{\frac{\mu}{\check{D}(t)}y_k(t)  -\nu(t)}$ at both side of \eqref{trans_equation}, we have
\begin{equation}\label{eq_1}
    -\nu(t)2^{\nu(t)} = \left(\frac{\mu}{\check{D}(t)}y_k(t)-\nu(t)\right)2^{-\left(\frac{\mu}{\check{D}(t)}y_k(t) -\nu(t)\right)}.
\end{equation}
Further we define $\tilde{y}_k(t) = \frac{\mu}{\check{D}(t)}y_k(t)-\nu(t)$, and multiply $\ln 2$ at both side of \eqref{eq_1}, we have
\begin{equation}
    \tilde{y}_k(t)\ln 2 \left(\exp^{\ln 2}\right)^{-\tilde{y}_k(t)} = -\ln 2\nu(t)2^{\nu(t)}.
\end{equation}
According to \cite{LambertW}, the solution for the above equation can be given via a Lambert-W function
\begin{equation}
    y_k(t) = -\frac{\check{D}(t){\rm LambW}\left(\nu(t)2^{\nu(t)}\ln{2} \right)}{\mu\ln 2}- \frac{\sigma^2}{|h_k(t)|^2},
\end{equation}
where ${\rm LambW}(\cdot)$ denotes the \textit{Lambert-W} function with $k=-1$.

\subsection{Proof of Theorem \ref{thm2}}\label{Proof_thm2}
We adopt the proof by contradiction with assuming that there is no threshold structure for any $\frac{\Delta}{2}-$optimal solution of Problem \textbf{P1}.
Let $\left\{\tilde{\mathcal{K}}^{\lozenge}(t), \{B_k^{\lozenge}(t)\}\right\}$ be the $\frac{\Delta}{2}-$optimal solution, then we can find a $n\leq K$ so that 
\begin{align}
			{\rm \rho}_j = \left\{\begin{array}{rc}
	0&\max\{1, k^\star - n\}<j<k^\star\\
	1&k^\star<j<\min\{K, k^\star +n\}
\end{array},\right.
\end{align}
where $\left\{{\rm idx_j}\;|\; \max\{1, k^\star - n\}<j<k^\star\right\}$ and $\{{\rm idx_j}\;|\;k^\star<j<\min\{K, k^\star +n\}\}$ can be swapped, and then the solution is a threshold structure.
Recall that the scheduling indexes beyond this region hold the same, the difference of $\mathcal{J}( \tilde{\mathcal{K}}(t), \{B_k(t)\})$ in terms of $\left\{\tilde{\mathcal{K}}^\star(t), \{B_k^\star(t)\}\right\}$ and $\left\{\tilde{\mathcal{K}}(t), \{B_k(t)\}\right\}$ is
\begin{equation*}
	\begin{split}
		&\mathcal{J}\left(\tilde{\mathcal{K}}^\star(t), \{B_k^\star(t)\}\right) - \mathcal{J}\left(\tilde{\mathcal{K}}(t), \{B_k(t)\}\right)\\
		=\;&\mathcal{J}\left(\tilde{\mathcal{K}}^\star(t), \{B_k^\star(t)\}\right) -\mathcal{J}\left(\tilde{\mathcal{K}}^{\lozenge}(t), \{B_k^{\lozenge}(t)\}\right) \\
		&+ \mathcal{J}\left(\tilde{\mathcal{K}}^{\lozenge}(t), \{B_k^{\lozenge}(t)\}\right)- \mathcal{J}\left(\tilde{\mathcal{K}}(t), \{B_k(t)\}\right)\\
		\leq\;&\frac{\Delta}{2}+\max\left\{\left||\tilde{\mathcal{K}}^{\lozenge}(t)|-|\tilde{\mathcal{K}}(t)|\right| - \zeta(t)\hat{D}(t)\left|D^{\lozenge}(t) - D(t)\right|\right\} \\
	\mathop{\leq}^{(b)}\;&\frac{\Delta}{2}+\max\left\{\left||\tilde{\mathcal{K}}^{\lozenge}(t)|-|\tilde{\mathcal{K}}(t)|\right| \right\}=\frac{\Delta}{2}+K=\Delta,
	\end{split}
\end{equation*}
where $(b)$ holds due to $-\zeta(t)\hat{D}(t)\left|D^{\lozenge}(t) - D(t)\right|<0$. To this end, we conclude that there is at least one solution is $\Delta-$optimal for Problem \textbf{P1} with the threshold structure. Note that the proposed online algorithm determines the scheduling set and the associated bandwidth allocation by sequentially comparing $\mathcal{J}\left(\tilde{\mathcal{K}}(t), \{B_k(t)\}\right)$ with the previously achieved one, the output of the proposed online algorithm is $\Delta-$ optimal.
\bibliographystyle{IEEEtran}
\bibliography{ref.bib}
\end{document}